\documentclass[a4paper,11pt,leqno]{amsart}
\usepackage{a4wide,color,array}
\usepackage{cite}
\usepackage{hyperref}
\usepackage{amsmath,amssymb,amsthm,color}
\hypersetup{linktocpage,colorlinks, citecolor={magenta}}
\usepackage[english]{babel}
\usepackage[utf8]{inputenc}

\usepackage{xspace}
\usepackage{bm}				% lettre math√©matiques grasses et
\usepackage{bbm,upgreek}		% lettres math blackboard et grecques pour Œºm et pour Œ≤-decay
\usepackage[e]{esvect}		% des fl√®ches de vecteurs plus √©l√©gantes
\usepackage{esint}			% int√©grales diverses; installer esint-type1
\usepackage{etoolbox}			% nombreuses fonctions avanc√©es indispensables
\usepackage{calc}				% calcul infix des longueurs
\usepackage{eso-pic}			% pour placer des √©l√©mnst das le background lors du shipout
\usepackage{datetime}			% formatage des datees, acc√®s au temps

\usepackage{lmodern}
\numberwithin{equation}{section}

\usepackage{enumitem}
\setlist{nosep}
\setlist{noitemsep}

\newcommand{\R}{\mathbb{R}}

\newtheorem{theo}{Theorem}
\newtheorem{prop}{Proposition}[section]
\newtheorem{lem}[prop]{Lemma}
\newtheorem{coro}[prop]{Corollary}

\newtheorem{defi}[prop]{Definition}

\theoremstyle{plain}

\def \t0{\rightarrow 0} % Vers z√©ro
\def \ti{\rightarrow \infty} % Vers l'infini

\def \hal{\frac{1}{2}}

\def \1{\mathbf{1}} % Fonction caract√©ristique

\def \ep{\varepsilon}

\def\nab{\nabla}

\def \KNbeta{K_{N,\beta}}

\def \PNbetaV{\mathbb{P}_{N, \beta}} % Mesure de Gibbs √† \beta

\def\XXint#1#2#3{{\setbox0=\hbox{$#1{#2#3}{\int}$}
     \vcenter{\hbox{$#2#3$}}\kern-.5\wd0}}

\def \XN{\vec{X}_N}

\def \FN{F_N}

\def \Fluct{\mathrm{Fluct}}
\def \fluct{\mathrm{fluct}}

\def \mueq{\mu_0} % Eq. measure
\def \HNV{\mathcal{H}_N}
\def\HNV{H_N^{\mueq}}

\def\Esp{\mathbf{E}} % Esp√©rance

\def \ZNbeta{Z_{N,\beta}}

\def \Ani{\mathsf{A}}

\def \Vt{V_{t}}

\def \id{\mathrm{Id}}

\def \mueq{\mu_{0}}
\def \mueqt{\mu_{t}}

\def \Ent{\mathrm{Ent}}

\def\indic{\mathbf{1}}

\def\drd{\delta_{\R}}

\def\Err{\mathsf{Error}}
\def\({\left(}
\def\){\right)}

\makeatletter
\def\namedlabel#1#2{\begingroup
    #2%
    \def\@currentlabel{#2}%
    \phantomsection\label{#1}\endgroup
}
\makeatother
\def \epsilon{\varepsilon}

\def \nn{\mathsf{n}}
\def \mm{\mathsf{m}}
\def \kk{\mathsf{k}}
\def \pp{\mathsf{p}}
\def\ll{\mathsf{l}}
\def \rr{\mathsf{r}}
\def \Rxd{R_{x,d}}
\def \HNV{\mathcal{H}_N^V}

\def \mueq{\mu_V}
\def \muV{\mu_V}
\def \mueqz{\mu_0}
\def \mueqt{\mu_t}
\def \zetaV{\zeta_V}

\def \tzetat{\tilde{\zeta}_t}
\def \PNbetaV{\mathbb{P}_{N, \beta}^V}
\def \ZNbetaV{Z_{N, \beta}^V}
\def\ZNbetaVt{Z_{N,\beta}^{V_t}}

\def \KNbeta{K_{N, \beta}}
\def \SigmaV{\Sigma_V}
\def \PV{P.V.}
\def \PNbeta{\mathbb{P}_{N, \beta}}
\def \cV{c_V}
\def \XiV{\Xi_{V}}

\def \Re{\mathcal{R}}
\def \Im{\mathcal{I}}
\def \SigmaVint{\mathring{\Sigma}_V}

\def \cxi{c_{\xi}}
\def \Ff{\mathcal{F}}
\def \comT{}

\begin{document}
\title{CLT for Fluctuations of $\beta$-ensembles with general potential}
\author{Florent Bekerman, Thomas Lebl\'e, and Sylvia Serfaty}
\date{\today}
\address[F. Bekerman]{Department of Mathematics, Massachusetts Institute of Technology, 77 Massachusetts Ave, Cambridge,
MA 02139-4307 USA.}
\email{bekerman@mit.edu}
\address[T. Lebl\'e]{Courant Institute, New York University, 251 Mercer st, New York, NY 10012,~USA.}
\email{thomasl@cims.nyu.edu}

\address[S. Serfaty]{ Courant Institute, New York University, 251 Mercer st, New York, NY 10012,~USA.}
\email{serfaty@cims.nyu.edu}

\begin{abstract} We prove a central limit theorem for the linear statistics of one-dimensional log-gases, or $\beta$-ensembles. We use a  method based on a change of variables which allows to treat fairly general situations, including multi-cut and, for the first time, critical cases, and generalizes the previously known results of Johansson, Borot-Guionnet and Shcherbina. In the one-cut regular case, our approach also allows to retrieve a rate of convergence as well as   previously known expansions of the free energy  to  arbitrary order.
\end{abstract}

\maketitle

{\bf keywords:} $\beta$-ensembles, Log Gas, Central Limit Theorem, Linear statistics.
%\tableofcontents
\section{Introduction}
Let $\beta > 0$ be fixed. For $N \geq 1$, we are interested in the $N$-point canonical Gibbs measure\footnote{We use $\frac{\beta}{2}$ instead of $\beta$ in order to match the existing literature. The first sum in \eqref{def:HN}, over indices $i \neq j$, is twice the physical one, but is more convenient for our analysis.} for a one-dimensional log-gas at the \textit{inverse temperature} $\beta$, defined by
\begin{equation}\label{def:PNbeta}
d\PNbetaV(\XN) = \frac{1}{\ZNbetaV} \exp \left( - \frac{\beta}{2} \HNV(\XN)\right) d\XN,
\end{equation}
where $\XN = (x_1, \dots, x_N)$ is an $N$-tuple of points in $\R$, and $\HNV(\XN)$, defined by
\begin{equation} \label{def:HN}
\HNV(\XN) := \sum_{1 \leq i \neq  j \leq N} - \log |x_i-x_j| + \sum_{i=1}^N N V(x_i),
\end{equation}
is the energy of the system in the state $\XN$, given by the sum of the pairwise repulsive logarithmic interaction between all particles plus the effect on each particle of an external field or confining potential $NV$ whose intensity is proportional to $N$. We will use $d\XN$ to denote the Lebesgue measure on $\R^N$. The constant $\ZNbetaV$ in the definition \eqref{def:PNbeta} is the normalizing constant, called the \textit{partition function}, and is equal to
\begin{equation*}
\ZNbetaV := \int_{\R^N} \exp \left( - \frac{\beta}{2} \HNV(\XN)\right) d\XN.
\end{equation*}
Such systems of particles with logarithmic repulsive interaction on the line have been extensively studied, in particular because of their connection with random matrix theory, see \cite{forrester} for a survey.  

Under mild assumptions on $V$, it is known that the empirical measure of the particles converges almost surely to some deterministic probability measure on $\R$ called the \textit{equilibrium measure} $\muV$, with no simple expression in terms of $V$. For any $N \geq 1$, let us define the \textit{fluctuation measure}
\begin{equation} \label{def:fluctN}
\fluct_N :=  \sum_{i=1}^N \delta_{x_i} - N \muV,
\end{equation}
which is a random signed measure. For any test function $\xi$ regular enough we define the \textit{fluctuations of the linear statistics associated to $\xi$} as the random real variable
\begin{equation} \label{def:FluctN}
\Fluct_N(\xi) := \int_{\R} \xi\, d\fluct_N.
 \end{equation}
The goal of this paper is to prove a Central Limit Theorem (CLT) for $\Fluct_N(\xi)$, under some regularity assumptions on $V$ and $\xi$. 

\subsection{Assumptions}
\begin{description}
\item[\namedlabel{H1}{(H1)} - Regularity  and growth of $V$]
 The potential $V$ is in $C^\pp(\mathbb{R})$ and satisfies the growth condition 
 \begin{equation} \label{growV}
 \liminf_{|x| \to \infty} \frac{V(x)}{2 \log |x|} > 1.
 \end{equation}
 \end{description}
It is well-known, see e.g. \cite{ST}, that if $V$ satisfies \ref{H1} with $\pp \geq 0$, then the \textit{logarithmic potential energy} functional defined on the space of probability measures by
\begin{equation}\label{energy}
\mathcal{I}_V(\mu) = \int_{\R\times \R} -\log |x-y|\, d\mu(x) \, d\mu(y) + \int_{\R} V(x) \, d\mu(x)
\end{equation}
has a unique global minimizer $\muV$, the aforementioned \textit{equilibrium measure}. This measure has a compact support that we will denote by $\SigmaV$, and $\muV$ is characterized by the fact that there exists a constant $\cV$ such that the function $\zetaV$ defined by
\begin{equation}\label{zeta}
\zetaV(x) :=  \int - \log |x-y|  d\muV(y) + \frac{V(x)}{2} - \cV
\end{equation}
satisfies the Euler-Lagrange conditions
\begin{equation} \label{zetaEL}
\zetaV \geq 0 \text{ in } \R, \quad \zetaV = 0 \text{ on } \SigmaV.
\end{equation}

We will work under two additional assumptions: one deals with the possible form of $\muV$ and the other one is a non-criticality hypothesis concerning $\zetaV$. 
\begin{description}
\item[\namedlabel{H2}{(H2)}  -  Form of the equilibrium measure] 
The support  $\SigmaV$ of  $\muV$  is a finite union of $\nn +1$ non-degenerate intervals  
$$
\SigmaV = \underset {{0\leq l \leq \nn}}{\bigcup} [\alpha_{l,-} ; \alpha_{l,+}], \text{ with $\alpha_{l,-} < \alpha_{l,+}$.}
$$  
The points $\alpha_{l, \pm}$ are called the \textit{endpoints} of the support $\SigmaV$. For $x$ in $\SigmaV$, we let
\begin{equation}
\label{def:petitsigma} \sigma(x) := \prod_{l=0}^\nn \sqrt{\lvert x - \alpha_{l,-} \rvert \lvert x - \alpha_{l,+} \rvert}.
\end{equation}
We assume that the equilibrium measure has a density with respect to the Lebesgue measure on $\SigmaV$ given by
\begin{equation} \label{formmu}
\muV(x)= S(x) \sigma(x),
\end{equation}
where $S$ can be written as
\begin{equation}\label{regularity}
S(x) = S_0(x) \prod_{i=1}^\mm (x-s_i)^{2 \kk_i}, \quad S_0 > 0 \text{ on } \SigmaV,
\end{equation}
where $\mm \geq 0$, all the points $s_i$, called \textit{singular points}\footnote{Let us emphasize that a singular point $s_i$ can be equal to an endpoint $\alpha_{l,\pm}$.}, belong to $\SigmaV$ and the $\kk_i$ are natural integers.
\item[\namedlabel{H3}{(H3)} - Non-criticality of $\zetaV$]
The function $\zetaV$ is  positive on $\R \setminus \SigmaV$.
\end{description}

\subsection{Main result}
\begin{defi}
\label{defi:Xi}
We introduce the so-called \textit{master operator} $\XiV$, which acts on $C^1$ functions by
\begin{equation} \label{Xi}
\XiV[\psi]:= - \hal \psi V' +\int\frac{\psi(\cdot)- \psi(y)}{\cdot -y}\, d\muV(y).
\end{equation}
\end{defi}

\begin{theo}[Central limit theorem for fluctuations of linear statistics] \label{theoreme}
Let $\xi$ be a  function in $C^{\rr}(\R)$, assume that \ref{H1}, \ref{H2}, \ref{H3} hold. We let 
$$
\kk = \max_{i=1, \dots, \mm} 2 \kk_i,
$$
where the $\kk_i$'s are as in \eqref{regularity}.
Assume that
\begin{equation}
\label{conditionpr} \pp \geq  (3\kk + 6), \quad \rr \geq (2\kk+4),
\end{equation}
where $\pp$ (resp. $\rr$) denotes the regularity of $V$ (resp. $\xi$)

If $\nn \geq 1$, assume that $\xi$ satisfies the $\nn$ following conditions
\begin{equation}\label{X1}
\int_{\SigmaV} \frac{\xi(y) y^d}{\sigma(y)} dy = 0 \qquad \text{ for }  d= 0 ,\dots, \nn  -1.
\end{equation}

Moreover, if $\kk \geq 1$, assume that for all $i = 1, \dots, \mm$
\begin{equation}
\label{X2}  
\int_{\SigmaV} \frac{\xi(y) - R_{s_i,d}\xi(y)}{\sigma(y)(y-s_i)^d} dy =0 \qquad \text{ for }  d= 1 ,\dots,  2\kk_i, 
\end{equation}
 where $\Rxd \xi$ is the Taylor expansion of $\xi$ to order $d-1$ around $x$ given by
\begin{equation*}
\Rxd \xi(y) = \xi(x) + (y-x)\xi'(x) + \cdots + \frac{(y-x)^{d-1}}{(d-1)!} \xi^{(d-1)}(x).
\end{equation*}

Then there exists a constant $c_\xi $ and a function $\psi$ of class $C^3$ in some open neighborhood $U$ of $\SigmaV$ such that $\XiV[\psi] = \frac{\xi}{2}+ c_{\xi}$ on $U$, and the fluctuation $\Fluct_N(\xi) $ converges in law as $N \to \infty$ to a Gaussian distribution with mean 
\begin{equation*}
m_\xi = \left( 1 - \frac{2}{\beta} \right)  \int {\psi'}\ d\muV,
\end{equation*}
and variance
\begin{equation*} 
v_\xi = - \frac{2}{\beta} \int {\psi \xi'} d\muV.
\end{equation*}
\end{theo}

It is proven in \eqref{identiteaprouver} that the variance $v_\xi$ has the equivalent expression
\begin{equation} \label{variance2}
v_{\xi} := \frac{2}{\beta} \left( \iint \left( \frac{\psi(x)-\psi(y)}{x-y} \right)^2 d\muV(x) d\muV(y)+ \int V'' \psi^2 d\muV\right).
\end{equation}
Let us note that $\psi$, hence also $m_{\xi}$ and $v_{\xi}$, can be explicitly written in terms of $\xi$.

\subsection{Comments on the assumptions}
The growth condition \eqref{growV} is standard and expresses the fact that the logarithmic repulsion is beaten at long distance by the confinement, thus ensuring that $\mueq$ has a compact support. Together with the non-criticality assumption \ref{H3} on $\zetaV$, it implies that the particles of the log-gas effectively stay within some neighborhood of $\SigmaV$, up to very rare events. 

The case $\nn = 0$, where the support has a single connected component, is called \textit{one-cut}, whereas $\nn \geq 1$ is a \textit{multi-cut} situation. If $\mm \geq 1$, we are in a \textit{critical case}. 

The relationship between $V$ and $\muV$ is complicated in general, and we mention some examples where $\muV$ is known to satisfy our assumptions.
\begin{itemize}
\item If $V$ is real-analytic, then the assumptions are satisfied with $\nn$ finite, $\mm$ finite and $S$ analytic on $\SigmaV$, see \cite[Theorem 1.38]{deift1998new}, \cite[Sec.1]{DKM}.
\item If $V$ is real-analytic, then for a “generic” $V$  the assumptions are satisfied with $\nn$ finite, $\mm = 0$ and $S$ analytic on $\SigmaV$, see \cite{kuijlaars2000generic}.
\item If $V$ is uniformly convex and smooth, then the assumptions are satisfied with $\nn = 0$, $\mm = 0$, and $S$ smooth on $\SigmaV$, see e.g. \cite[Example 1]{BdMS}. 
\item Examples of multi-cut, non-critical situations with $\nn = 0, 1, 2$ and $\mm = 0$, are mentioned in \cite[Examples 3-4]{BdMS}. 
\item An example of criticality at the edge of the support is given by 
$V(x) = \frac{1}{20} x^4 - \frac{4}{15} x^3 + \frac{1}{5} x^2 + \frac{8}{5} x$, 
for which the equilibrium measure, as computed in \cite[Example 1.2]{claeys2010higher}, is given by
$$
\muV(x) = \frac{1}{10 \pi} \sqrt{|x - (-2)||x -2|} (x-2)^2 \mathbf{1}_{[-2,2]}(x). 
$$
\item An example of criticality in the bulk of the support is given by 
$V(x) = \frac{x^4}{4} - x^2$, for which the equilibrium measure, as computed in \cite{claeys2006universality}, is 
$$
\muV(x) = \frac{1}{2\pi} \sqrt{|x - (-2)||x -2|} (x-0)^2 \mathbf{1}_{[-2,2]}(x).
$$
\end{itemize}
Following the terminology used in the literature \cite{DKM, kuijlaars2000generic, claeys2006universality}, we may say that our assumptions allow the existence of singular points of type II (the density vanishes in the bulk) and III (the density vanishes at the edge faster than a square root). Assumption \ref{H3} rules out the possibility of singular points of type I, also called “birth of a new cut”, for which the behavior might be quite different, see \cite{claeys2008birth, mo2008riemann}.

\subsection{Existing literature, strategy and perspectives}
\subsubsection{Connection to previous results}
The CLT for fluctuations of linear statistics in the context of $\beta$-ensembles was proven in the pioneering paper \cite{Johansson} for polynomial potentials in the case $\nn = 0, \mm=0$, and generalized in \cite{scherbi1} to real-analytic potentials in the possibly multi-cut, non-critical cases ($\nn \geq 0, \mm = 0$), where a set of $\nn$ necessary and sufficient conditions on a given test function in order to satisfy the CLT is derived. If these conditions are not fulfilled, the fluctuations are shown to exhibit oscillatory behaviour. Such results are also a by-product of the all-orders expansion of the partition function obtained in \cite{BorGui1} ($\nn =0, \mm =0$) and \cite{BorGui2} ($\nn \geq 0, \mm = 0$). A CLT for the fluctuations of linear statistics for test functions living at mesoscopic scales was recently obtained in \cite{bekerman2016}.
Finally, a new proof of the CLT in the one-cut non-critical case was very recently given in \cite{Lambert:2017kq}. It is based on Stein's method and provides a rate of convergence in Wasserstein distance.
\subsubsection{Motivation and strategy}
Our goal is twofold: on the one hand, we provide a simple proof of the CLT using a change of variables argument, retrieving the results cited above. On the other hand, our method allows to substantially relax the assumptions on $V$, in particular for the first time we are able to treat critical situations where $\mm \geq 1$.

Our method, which is adapted from the one introduced in \cite{ls2} for two-dimensional log-gases, can be summarized as follows
\begin{enumerate}
\item We prove the CLT by showing that the Laplace transform of the fluctuations converges to the Laplace transform of the correct Gaussian law. This idea is already present in \cite{Johansson} and many further works. Computing the Laplace transform of $\Fluct_N(\xi)$ leads to working with a new potential $V + t \xi$ (with $t$ small), and thus to considering the associated perturbed equilibrium measure.
\item Following \cite{ls2}, our method then consists in finding a change of variables (or a transport map) that pushes $\muV$ onto the perturbed equilibrium measure. In fact we do not exactly achieve this, but rather we construct a transport map  $I + t\psi$, which is a perturbation of identity, and consider the \textit{approximate} perturbed equilibrium measure $(I+t\psi) \# \muV$. The map $\psi$ is found by inverting the operator \eqref{Xi}, which is well-known in this context, it appears e.g. in \cite{BorGui1, BorGui2, scherbi1, BFG}. A CLT will hold if the function $\xi$ is (up to constants) in the image of $\XiV$, leading to the conditions \eqref{X1}--\eqref{X2}. The  change of variables approach for one-dimensional log-gases was already used e.g. in \cite{shchange, BFG}, see also \cite{guionnet2007second,guionnet2014free} which deal with the non-commutative context.
\item The proof then leverages on the expansion of $\log \ZNbetaV$ up to order $N$ proven in \cite{ls1},  valid in the multi-cut and critical case, and whose dependency in $V$ is explicit enough. This step replaces the a priori bound on the correlators used e.g. in \cite{BorGui1}.
\end{enumerate}

\subsubsection{More comments and perspectives}
Using the Cramér-Wold theorem, the result of Theorem \ref{theoreme} extends readily to any finite family of test functions satisfying the conditions (\eqref{X1}, \eqref{X2}): the joint law of their fluctuations converges to a Gaussian vector, using the bilinear form associated to \eqref{variance2} in order to determine the covariance.

In the multi-cut case, the CLT results of \cite{scherbi1} or \cite{BorGui2} are stated under $\nn$ necessary and sufficient conditions on the test function, and the non-Gaussian nature of the fluctuations if these conditions are not satisfied is explicitly described. In the critical cases, we only state sufficient conditions \eqref{X2} under which the CLT holds. It would be interesting to prove that these conditions are necessary, and to characterize the behavior of the fluctuations for functions which do not satisfy \eqref{X2}.

Finally, we expect Theorem \ref{theoreme} to hold also at mesoscopic scales. \comT{The proof of \cite{bekerman2016} uses the rigidity estimates of \cite{bourgade2014} which are, to the best of our knowledge, not available to  the critical case.}

\subsection{The one-cut noncritical case}
In the case $\nn=0$ and $\mm =0$, following the transport approach, we can obtain the convergence of the Laplace transform of fluctuations with an explicit rate, under the assumption that $\xi$ is very regular (we have not tried to optimize in the regularity):
\begin{theo}[Rate of convergence in the one-cut noncritical case]\label{th2}
Under the assumptions of Theorem \ref{theoreme}, if in addition $\nn=0$, $\mm=0$, $\pp \ge 6 $ and $\rr \ge 18$, then we also have, for any $s$ such that $\frac{|s|}{N}$ is small enough\footnote{Depending only on $\xi$.} 
\begin{multline}\label{cvrate}
\left|\log \Esp_{\PNbetaV}\left[ \exp(s\Fluct_N(\xi))\right] + \left(1-\frac{\beta}{2}\right) \frac{2s}{\beta} \int \psi' d\muV+ \frac{s^2}{\beta} \int \xi' \psi d\muV \right|\\ =  \frac{s}{N} O\left( \|\psi\|_{C^{29}(U)} + \|\psi\|_{C^{5}(U)} + \|\psi\|^3_{C^{5}(U)}  + s \|\psi\|_{C^3}^2 + \sqrt{N} \| \psi \|_{C^2} \right).
\end{multline}
where the constant $C$ depends only on $V$ and $\beta$.
\end{theo}
\comT{These additional assumptions} allow to avoid  using the result  of \cite{ls1} on the expansion of $\log \ZNbetaV$.
 Our transport approach also provides a functional relation on the expectation of fluctuations which  allows by a boostrap procedure to recover an expansion of $\log \ZNbetaV$ (relative to a reference potential) to arbitrary powers of $1/N$ in very regular cases, i.e the result of \cite{BorGui1} but without the analyticity assumption.
 All these results are presented  in Appendix \ref{appb}.

\subsection{Some notation}
We denote by $\PV$ the principal value of an integral having a singularity at $x_0$, i.e. 
\begin{equation}
\label{def:V}
\PV \int f= \lim_{\ep\to 0}\int_{ -\infty}^{x_0-\ep}f+\int_{x_0+\ep}^{+\infty}f.
\end{equation}

If $\Phi$ is a $C^1$-diffeomorphism and $\mu$ a probability measure, we denote by $\Phi \# \mu$ the push-forward of $\mu$ by $\Phi$, which is by definition such that for $A \subset \R$ Borel,
$$
(\Phi \# \mu) (A) := \mu(\Phi^{-1}(A)).
$$

If $A \subset \R$ we denote by $\mathring{A}$ its interior.

For $k \geq 0$, and $U$ some bounded domain in $\R$, we endow the spaces $C^k(U)$ with the usual norm
$$
\| \psi \|_{C^k(U)} := \sum_{j=0}^k \sup_{x \in U} |\psi^{(j)}(x)|.
$$

If $z$ is a complex number, we denote by $\Re(z)$ (resp. $\Im(z)$) its real (resp. imaginary) part. 

For any probability measure $\mu$ on $\R$ we denote by $h^{\mu}$ the logarithmic potential generated by $\mu$, defined as the map
\begin{equation} \label{logpotential}
x \in \R^2  \mapsto h^{\mu}(x) = \int -\log |x-y| d\mu(y).
\end{equation}
\\

{\bf Acknowledgments:} We would like to thank Alice Guionnet for suggesting the problem and for helpful discussions.

\section{Next order energy and concentration bounds}  \label{sec:nextorder}
We start with the energy splitting formula of \cite{SS15a} that separates fixed leading order terms from variable next order ones, and allows to quickly obtain first concentration bounds. 
 \subsection{The next-order energy}
 For any probability measure $\mu$, let us define, 
\begin{equation}\label{energydef}
F_N(\vec{X}_N, \mu)  = -\iint_{(\R\times \R) \setminus \triangle} \log|x-y| \Big( \sum_{i=1}^N \delta_{x_i}-\mu\Big) (x) \Big( \sum_{i=1}^N \delta_{x_i}-\mu\Big) (y),
\end{equation}
where $\triangle$ denotes the diagonal in $\R \times \R$. 

We have the following splitting formula for the energy, as introduced in \cite{SS15a} (we recall the proof in Section \ref{sec:preuvesplitting}).
 \begin{lem} \label{lem:splitting}
 For any $\XN\in \R^N$, it holds that 
\begin{equation}\label{splitting}
\HNV (\vec{X}_N) = N^2 \mathcal{I}_V(\muV) + 2N \sum_{i=1}^N \zetaV(x_i) + F_N(\vec{X}_N, \muV) \ .
\end{equation}
\end{lem}

Using this splitting formula \eqref{splitting}, we may re-write $\PNbetaV$ as
\begin{equation} \label{PNbetaVdeux}
d\PNbetaV(\XN) = \frac{1}{\KNbeta(\muV, \zetaV)} \exp\left( - \frac{\beta}{2} \left( \FN(\XN, \muV) + 2N \sum_{i=1}^N \zetaV(x_i) \right) \right) d\XN,
\end{equation}
with a next-order partition function $\KNbeta(\muV, \zetaV)$ defined by
\begin{equation} \label{def:KNbeta}
\KNbeta(\muV, \zetaV) := \int_{\R^N} \exp\left( - \frac{\beta}{2} \left( \FN (\XN, \muV) + 2N \sum_{i=1}^N \zetaV(x_i) \right) \right) d\XN.
\end{equation}
We extend this notation to $\KNbeta(\mu, \zeta)$ where $\mu$ is a probability density and $\zeta$ is a confinement potential.

In view of \eqref{splitting}, we have 
\begin{equation}\label{ZK}
\ZNbetaV= \exp\( - \frac{\beta}{2} \mathcal{I}_V(\muV)\)\KNbeta(\muV,\zetaV).\end{equation}

\subsection{Expansion of the next order partition function}
If $\mu$ is a probability density, we denote by $\Ent(\mu)$ the entropy function given by\footnote{The sign convention here differs from the usual one.}
\begin{equation}
\label{def:Entmu}
\Ent(\mu) : = \int_{\R} \mu \log \mu.
\end{equation}
The following asymptotic expansion is proven \cite[Corollary 1.1]{ls1} (cf. \cite[Remark 4.3]{ls1}) and valid in a general multi-cut critical situation. 
\begin{lem} \label{lem:comparaison} Let $\mu$ be a probability density on $\R$. Assume that $\mu$ has the form \eqref{formmu}, \eqref{regularity} with $S_0$ in $C^2(\Sigma)$, and that $\zeta$ is some Lipschitz function on $\R$ satisfying 
$$
\zeta = 0 \text{ on } \Sigma, \quad \zeta > 0 \text{ on } \R \setminus \Sigma,\quad  \int_{\R} e^{-\beta N \zeta(x)} dx<\infty \ \text{for $N$ large enough}.
$$ 
Then, with the notation of \eqref{def:KNbeta} and for some $C_\beta$ depending only on $\beta$, we have
\begin{equation}\label{expzformula}
\log \KNbeta(\mu, \zeta) = \frac{\beta}{2} N\log N + C_\beta N- N \left(1 - \frac{\beta}{2}\right) \Ent(\mu)  + N o_N(1).
\end{equation}
\end{lem}

\subsection{Exponential moments of the energy and the fluctuations}
\comT{In this paragraph we show that the next-order energy is typically (in a strong sense) of order at most $N$, and that the fluctuations of a function in $C^1_c(\R)$ are of order at most $\sqrt{N}$.}
\subsubsection{Exponential moments of the next-order energy}
\begin{lem} We have, for some constant $C$ depending on $\beta$ and $V$
\begin{equation}\label{claimeq}
\left| \log \Esp_{\PNbetaV} \left[ \exp\left(\frac{\beta}{4} \left(F_N(\XN, \muV)   +N\log N \right) \right)\right] \right|\le C N.
\end{equation}
\end{lem}
\begin{proof}
This follows e.g. from \cite[Theorem 6]{SS15a}, but we can also deduce it from Lemma \ref{lem:comparaison}. We may write   
\begin{multline*}
\Esp_{\PNbetaV} \left[ \exp\left(\frac{\beta}{4} F_N(\XN, \muV) \right)\right] \\ = \frac{1}{\KNbeta(\muV, \zetaV)} \int \exp \left( -\frac{\beta}{4} \left( F_N(\XN, \muV) -  2N \sum_{i=1}^N 2 \zeta_V(x_i)\right) \right) d\XN 
\\
=  \frac{K_{N, \frac{\beta}{2}}(\muV, 2\zetaV)}{\KNbeta(\muV, \zetaV)}.
\end{multline*} 
Taking the $\log$ and using \eqref{expzformula} to expand both terms up to order $N$ yields the result.
\end{proof} 

\subsubsection{The next-order energy controls the fluctuations}  
 The following result is a consequence of the analysis of \cite{SS15a,PS15}, we give the proof in Section \ref{sec:preuvefluctenergy} for completeness. It shows that $F_N$ controls $\fluct_N$. 
  \begin{prop}\label{fluctenergy}
 If $\xi$ is  compactly supported and Lipschitz, we have, for some universal constant $C$
 \begin{multline}
 \label{controlfluct}
 \left|\int\xi\, d\fluct_N \right|
 \\ \le   \| \xi'\|_{L^\infty} +  (\| \xi'\|_{L^2}+\|\xi\|_{L^2})  \left( F_N(\vec{X}_N, \muV) + N\log N   + C(\|\muV\|_{L^\infty}+1) N\right)^{1/2}.
 \end{multline}
 \end{prop}
Combining this result with Lemma \ref{claimeq} and using H\"older's inequality, we deduce the following concentration result, improving on the previous concentration estimates in $\sqrt{N \log N}$ of \cite{BorGui1,MR3201924}.
\begin{coro}[Exponential moments of the fluctuations]
\label{cor:concentration}
For any  $\xi$ compactly supported and Lipschitz function, if $ \|\xi \|_{H^1(\R)} $ is small enough depending on $\beta$,  we have
\begin{equation}\label{concentration}
\log \Esp_{\PNbetaV} \left[ \exp \left(  \Fluct_N(\xi)\right) \right] \le C  \sqrt N \left(\| \xi' \|_{L^2(\R)}+\|\xi\|_{L^2(\R)}) \right)+ C \|\xi'\|_{L^\infty(\R)}
\end{equation}
where $C$ depends on $\beta $ and $V$.
\end{coro}
\comT{In view of the CLT result, one would expect to find concentration bounds in terms of the $H^{1/2}$ norm of $\xi$, but we do not pursue this goal here.}

\subsubsection{Confinement bound}
We will also need the following bound on the confinement. This is a  well-known fact, an easy  proof can for instance be given by following  that  of Lemma 3.3 of \cite{ls2}.
\begin{lem}\label{confinement} For any  fixed open neighborhood $U$ of $\Sigma$,
$$
{\PNbetaV} \left(  \XN \in U^N \right) \geq 1 - \exp(-c N)
$$
where $c>0$  depends on $U$ and $\beta$.
\end{lem}
Lemma \ref{confinement} is the only place where we use the non-degeneracy assumption \ref{H3} on the next-order confinement term $\zetaV$.

\section{Inverting the operator and defining the approximate transport}
The goal of this section  is to find transport maps $\phi_t$ for  $t$ small enough such that the transported measure $ \phi_t \# \mueqz$ approximates the equilibrium measure associated to $V_t:=V+t\xi$.   Since the equilibrium measures are characterized by  \eqref{zeta} with equality on the support, it is natural to search for  $\phi_t$ such that the quantity 
\begin{equation*}
\int -\log |\phi_t(x)-\phi_t(y)| d\mueqz(y) +  \frac{1}{2}\Vt(\phi_t(x))
\end{equation*}
is close to a constant. This is directly related to inverting the  operator $\Xi_V$ of \eqref{Xi},  and we will see that this choice  allows to cancel out some crucial terms later.

\subsection{Preliminaries}
\begin{lem}
\label{lemmareg} 
We have the following
\begin{itemize}
\item The function $S_0$ of \eqref{regularity} is in $C^{\pp - 3 - 2\kk}(\SigmaV)$. 
\item There exists an open neighborhood $U$ of $\SigmaV$ and a \comT{positive }function $M$ in $C^{\pp-3-2\kk}(U \setminus \SigmaVint)$ such that
\begin{equation}\label{factorization}
\zetaV'(x)  = M(x) \sigma(x) \prod_{i=1}^m (x-s_i)^{2 \kk_i}.
\end{equation}
\end{itemize}
\end{lem}
In particular, \eqref{factorization} quantifies how fast $\zetaV'$ vanishes near an endpoint of the support. We postpone the proof to Section \ref{sec:preuvelemmareg}.

\subsection{The approximate equilibrium measure equation}
In the following, we let
\begin{itemize}
\item $U$ be an open neighborhood of $\SigmaV$ such that \eqref{factorization} holds.
\item $B$ be the open ball of radius $\hal$ in $C^2(U)$.
\end{itemize}

We define a map $\Ff$ from $[-1,1] \times B$  to $C^{1}(U)$ by setting $\phi := \id + \psi$ and
\begin{equation}\label{defF}
\Ff(t,\psi) : = \int - \log |\phi(\cdot) - \phi(y) | d\mueq(y) + \frac{1}{2} V_t\circ \phi ( \cdot ) \ ,
\end{equation}

\begin{lem}\label{lem32}
The map $\Ff$ takes values in $C^{1}(U)$ and has continuous partial derivatives in both variables. Moreover there exists $C$ depending only on $V$ such that for all $(t,\psi)$ in $[-1 , 1] \times B$ we have
\begin{equation}\label{erreur transp}
\left\|\Ff(t, \psi) - \Ff(0,0) - \frac{t}{2} \xi + \XiV[\psi] \right\|_{C^{1}(U)} \leq C t^2  \|\psi\|^2_{C^2(U)},
\end{equation}
where $\XiV$ is as in \eqref{Xi}.
\end{lem}
The proof is postponed to Section \ref{sec:preuveF32}.

\subsection{Inverting the operator} \label{inversion}
\begin{lem}\label{leminv}
Let $\psi$ be defined by
\begin{align}
\label{psionSigma}
& \psi(x) = -\displaystyle\frac{1}{2 \pi^2 S(x) } \left(\displaystyle \int_{\Sigma}\displaystyle \frac{\xi(y) - \xi(x)}{\sigma(y)(y-x)} dy \right) 
\quad \text{ for } x \text{ in } \SigmaV, \\
\label{psioutSigma} & \psi(x) = \displaystyle \frac{ \displaystyle \int \displaystyle \frac{\psi(y)}{x-y} d\muV(y) +\displaystyle  \frac{\xi(x)}{2} + \cxi}{\displaystyle \int \displaystyle\frac{1}{x-y} d\muV(y) - \frac{1}{2} V'(x)} \quad \text{for } x \in U\backslash \SigmaV,
\end{align}
then $\psi$ is in $C^{\ll}(U)$  with $\ll=(\pp-3-3\kk)\wedge (\rr-1-2\kk)$ and 
\begin{equation}\label{contpsixi}
\|\psi\|_{C^{\ll}(U)} \le C \|\xi\|_{C^{\rr}(\R)}
\end{equation} for some constant $C$ depending only on $V$,
and there exists a constant $\cxi$ such that 
$$
\XiV[\psi]= \frac{\xi}{2}+\cxi \ \text{in} \ U,
$$
with $\XiV$ as in \eqref{Xi}.
\end{lem}
The proof of Lemma \ref{leminv} is postponed to Section \ref{sec:proofinverse}. In view of our assumptions, $\psi$ is  in $C^3(U)$ and we may extend it to $\mathbb{R}$ in such a way that it  is in $C^{3}(\R)$ with compact support.

\subsection{Transport and approximate equilibrium measure}
\def\psit{\psi_t}
\def\phit{\phi_t}
\def\tct{\tilde{c}_t}
\def\tmut{\tilde{\mu}_t}
\def\tzetat{\tilde{\zeta}_t}
\def\PNbetat{\PNbeta^{(t)}}
\def\taut{\tau_t}
\def \tm{t_{\max}}
We let $\psi$ be the function defined in Lemma \ref{leminv}, and $\cxi$ be such that
$$
\XiV[\psi] = \frac{\xi}{2} + \cxi \text{ on } U.
$$
We let 
\begin{equation}
\label{def:tmax} \tm : = \left(2 \|\psi\|_{C^2(U)}\right)^{-1},
\end{equation}
\begin{defi} \label{definitionsapprochees} For $t \in [-\tm,\tm]$, 
\begin{itemize}
\item We let $\psit$ be given by $\psit :=  t  \psi.$ 
\item We let $ \tct := t \cxi$.
\item We let $\phit$ be the \textit{transport}, defined by $\phit := \id + \psit.$
\item We let $\tmut$ be the \textit{approximate equilibrium measure}, defined by $\tmut :=  \phit \# \muV.$
\item We let $\tzetat$ be the \textit{approximate confining term} $\tzetat := \zetaV \circ \phit^{-1}$.
\end{itemize}
Finally, we let $\taut$ be defined by 
\begin{equation}
\label{def:taut}
\taut := \Ff(t,\psi_t) - \Ff(0,0)  - \tct.
\end{equation} 

\end{defi}

\begin{lem} Under our assumptions, the following holds 
\begin{itemize}
\item The map $\psit$ satisfies
$$
\XiV[\psit] = \frac{t}{2} \xi + \tct.
$$
\item The map $\phit$ is a $C^2$-diffeomorphism which coincides with the identity outside a compact support independent of $t \in [-\tm, \tm]$.
\item The error $\taut$ is a $O(t^2)$, more precisely
\begin{equation} 
\label{318}  \|\taut\|_{C^1(U)} \leq Ct^2 \|\psi\|^2_{C^2(U)} .\end{equation}
%\label{eqref319bis}&  \|\hat \tau_t\|_{C^1(\R^2)} \le C t^2\|\psi\|_{C^2(U)}^2.

%\item On $\phi_t(\Sigma_V)$, we have
%\begin{equation} \label{tzetategal}
%\tzetat = h^{\tmut} + \frac{V_t}{2} - \tct - c_V - \tau_t \circ \phi_t^{-1}. 
%\end{equation}
\end{itemize}
\end{lem}
\begin{proof} The first two points are straightforward, the bound \eqref{318} follows from \eqref{erreur transp} and the definitions.
\end{proof}

\comT{In the sequel, we will use the fact that the result of Lemma \ref{confinement} allows us to assume  that the points of $\XN$  all belong to the neighborhood  $U$ for $t$ small enough, except
 for an event of exponentially small probability.}

\section{Study of the Laplace transform}
We now follow the standard approach of reexpressing the Laplace transform of fluctuations in terms of a ratio of partition functions, and combine it with the change of variables approach, in the following central computation.
\def \Const{\mathsf{Const}}

\subsection{Expansion of the Laplace transform of the fluctuations}
\begin{prop}\label{fonda}
Let $s$ be in $\R$, let  $t := \frac{-2s}{\beta N}$, and assume that $|t| \leq \tm$. We have
\begin{equation}\label{Laplace0}
\Esp_{\PNbetaV}\left[ \exp\left( s \Fluct_N(\xi) \right) \right] 
= \exp\(-s N\int \xi d\muV\) \frac{\ZNbetaVt}{\ZNbetaV}
\end{equation}
and 
\begin{multline}\label{Laplace}\Esp_{\PNbetaV}\left[ \exp\left( s \Fluct_N(\xi) \right) \right] 
\\ = \exp\left(\Const\right) 
\Esp_{\PNbetaV} \left( \exp\left(t \frac{\beta}{2} \Ani[\XN,\psi] +\left(1-\frac{\beta}{2}\right) \int \log \phi_t' d\fluct_N +  \Err\right)\right)
\end{multline}
where we define 
\begin{align}
\label{formeM}
\Const & =  -\frac{\beta}{4}  N^2 t^2 \int \xi'\psi d\muV +tN \left(1-\frac{\beta}{2}\right) \int \psi' d\muV, \\
\label{defAni}
\Ani[\XN,\psi] & = \iint_{\R\times \R} \frac{\psi(x)-\psi(y)}{x-y} d\fluct_N(x)d\fluct_N(y).
\end{align}
The $\Err$ term satisfies, for any fixed $u$
\begin{equation}
\label{formeerr}
\left|\log \Esp_{\PNbetaV} \left(\exp(u \Err)\right)\right| \le C_u \left(t^2 N \sqrt{N} \|\psi\|_{C^3}^2 + t^3 N^2 \| \psi \|_{C^1}^3\right).
\end{equation}
\end{prop}

To prove this result, we will use some auxiliary computations, whose proof is in Appendix~\ref{auxi}.
\begin{lem}\label{lem42}
For any bounded continuous function $h$ we have 
\begin{equation}\label{equil}
\iint \frac{h(x)-h(y)}{x-y}d\muV(x) d\muV(y) = \int V'(x) h(x) d\muV(x).
\end{equation}
For $\psi$ defined in Lemma \ref{leminv}, we have 
\begin{equation} \label{identiteaprouver}
 \int \xi' \psi d\muV  = - \iint \left(\frac{\psi(x) - \psi(y)}{x-y} \right)^2 d\muV(x) d\muV(y) -  \int V'' \psi^2 d\muV.
\end{equation}\end{lem}

\begin{proof}[Proof of Proposition \ref{fonda}]
By definition of  $\Fluct_N$ and in view of \eqref{PNbetaVdeux} we have 
\begin{equation}\label{Laplace00}
\Esp_{\PNbetaV}\left[ \exp\left( s \Fluct_N(\xi) \right) \right] 
=\frac{e^{-s N\int \xi d\muV}}{\ZNbetaV}
 \int_{\R^N} \exp\left( -\frac{\beta}{2}\left( \mathcal H_N(\XN)+ Nt\sum_{i=1}^N \xi(x_i)\right)\right) d\XN
 \end{equation} and \eqref{Laplace0}  immediately follows by definition of $V_t= V+t\xi$.
 
Let us now make the change of variables $x_i= \phi_t(y_i)$ with $\phi_t=\id + t\psi$ where $\psi$ is the map given in Lemma \ref{leminv}.
We obtain 
\begin{multline}\label{Laplace11}
 e^{s  N\int \xi d\muV}  \Esp_{\PNbetaV}\left[ \exp\left( s \Fluct_N(\xi) \right) \right] 
\\ =\frac{1}{\ZNbetaV} \int\exp\(-\frac{\beta}{2}\(-\sum_{i\neq j} \log |\phi_t(x_i)-\phi_t(x_j)| + N \sum_{i=1}^N (V+t\xi)(\phi_t(x_i))\right) + \sum_{i=1}^N \log \phi_t'(x_i)\right) \, d\XN\\
 = \Esp_{\PNbetaV} \left[\exp\left(-\frac{\beta}{2} \left( - \sum_{i\neq j} \log\frac{|\phi_t(x_i)-\phi_t(x_j)|}{|x_i-x_j|} 
+ N\sum_{i=1}^N \left( V_t(\phi_t(x_i))- V(x_i) \right)- \frac{2}{\beta} \sum_{i=1}^N \log \phi_t'(x_i) \right)\right)\right].
\end{multline}
Let us now focus on the exponent in the right-hand side. First, since $\psi$, hence $\phi_t$,  is $C^1$ we may  reinsert the diagonal terms and write 
\begin{multline}\label{lap12}
 - \sum_{i\neq j} \log\frac{|\phi_t(x_i)-\phi_t(x_j)|}{|x_i-x_j|} 
+ N\sum_{i=1}^N \left( V_t(\phi_t(x_i))- V(x_i)\right) -\frac{2}{\beta} \sum_{i=1}^N \log \phi_t'(x_i)
\\ =  - \sum_{i, j} \log\frac{|\phi_t(x_i)-\phi_t(x_j)|}{|x_i-x_j|} 
+ N\sum_{i=1}^N \left(V_t(\phi_t(x_i))- V(x_i)\right) + \left(1-\frac{2}{\beta}\right) \sum_{i=1}^N \log \phi_t'(x_i).
\end{multline}
Expanding around $N\muV$, we may next write 
\begin{multline}\label{expan}
 - \sum_{i, j} \log\frac{|\phi_t(x_i)-\phi_t(x_j)|}{|x_i-x_j|} 
+ N\sum_{i=1}^N \left(V_t(\phi_t(x_i))- V(x_i)\right) + \left(1-\frac{2}{\beta}\right) \sum_{i=1}^N \log \phi_t'(x_i)
=   T_2+ T_1+T_0
\end{multline}
where $T_0, T_1, T_2$ are as follows
\begin{align}\label{T2}
T_2 & =  -N^2\iint \log \frac{|\phi_t(x)-\phi_t(y)|}{|x-y|} \, d\muV (x) d\muV(y) + N^2\int( V_t \circ \phi_t -V) d\muV\\
\nonumber &+N   \left(1-\frac{2}{\beta}\right) \int \log \phi_t' d\muV \\
\label{T1} 
T_1 & =  -  2N  \iint \log \frac{|\phi_t(x)-\phi_t(y)|}{|x-y|} \, d\muV (x) d\fluct_N(y) + N  \int\left( V_t\circ \phi_t- V\right) d\fluct_N\\
&\nonumber +  \left(1-\frac{2}{\beta}\right) \int \log \phi_t'\, d\fluct_N \\
\label{415}
T_0 &  = - \iint  \log \frac{|\phi_t(x)-\phi_t(y)|}{|x-y|} d\fluct_N(x) d\fluct_N(y).
\end{align}

Next, we note  that the $T_2$ term is independent of the configuration, and we Taylor expand it as $t\to 0$ using that $\phi_t=\id + t\psi$. We may write that 
\begin{equation}
\label{logex}
\log \frac{|\phi_t(x)-\phi_t(y)|}{|x-y|}= t \frac{\psi(x)-\psi(y)}{x-y}-\frac{t^2}{2} \( \frac{\psi(x)-\psi(y)}{x-y}\)^2+ t^3 \ep_t(x,y)
\end{equation}
with $\|\ep_t(x,y)\|_{\comT{L^{\infty}}(\R\times \R)}\le C \|\psi\|_{C^1}^3$ and expand all other terms
to find 
\begin{multline}
\frac{T_2}{N^2}=- t \iint \frac{\psi(x)-\psi(y)}{x-y} d\muV(x) d\muV(y)+  \frac{t^2}{2} \iint  \( \frac{\psi(x)-\psi(y)}{x-y}\)^2d\muV(x)d\muV(y)\\
\\ + t \int V' \psi \,   d\mueq + \frac{t^2}{2} \int V'' \psi^2 d\mueq + t\int \xi d\mueq + t^2 \int \xi' \psi d\muV + \frac{t}{N}\left(1-\frac{2}{\beta}\right)\int \psi' d\muV
\\+ O\left(t^3 \|\psi\|_{C^1}^3  \comT{+ t^3 \|\psi\|_{L^{\infty}}^2\|\xi\|_{C^2}} + \frac{t^2}{N} \|\psi\|_{C^1}^2\right) .
\end{multline}

Applying \eqref{equil} to $\psi$ and using \eqref{identiteaprouver}, we find
\begin{multline}
\label{vT2}
T_2= N^2 t \int \xi d\muV + \hal N^2 t^2 \int \xi'\psi d\muV +  tN \left(1-\frac{2}{\beta}\right) \int \psi' d\muV
\\
+ O\left(t^3 N^2 \|\psi\|_{C^1}^3 \comT{+ t^3 N^2 \|\psi\|_{L^{\infty}}^2\|\xi\|_{C^2}} + t^2 N \|\psi\|_{C^1}^2\right).
\end{multline}
 
  We turn next to the $T_1$ term, which can be rewritten\footnote{This uses crucially the fact that $\psi$ is chosen to satisfy $\Xi_V(\psi)= \frac{\xi}{2}+c_\xi.$} in view of \eqref{def:taut} as 
\begin{equation}
\label{t1t}
T_1 = \int \left( 2N \tau_t +  \left(1-\frac{2}{\beta} \right)  \log \phi_t'\right)  d\fluct_N = \Fluct_N[2N \tau_t] + \int \left(1-\frac{2}{\beta}\right)  \log \phi_t'  d\fluct_N.
\end{equation}
with $\|\tau_t\|_{\comT{C^1}(U)} \le C t^2 \|\psi\|_{C^2(U)}^2$ as in \eqref{318}.
Thus, using Corollary \ref{cor:concentration} we get for any fixed $u$
\begin{equation}\label{relT1} 
\left|\log \Esp_{\PNbetaV} \left[ \exp\left( u  \Fluct_N[2N \tau_t] \right) \right]\right|\le C_u t^2 N \sqrt N \|\psi\|_{C^2}^2.
\end{equation}

For the $T_0$ term, we use \eqref{logex} to write 
\begin{equation}\label{T0}
T_0= \comT{-} t \Ani[\XN,\psi] + t^2 \int \ep(x,y) d\fluct_N(x) d\fluct_N(y)
\end{equation}
with $\|\ep\|_{C^2(\R \times \R)}\le C \|\psi\|_{C^3}^2.$

Applying the result of Proposition \ref{fluctenergy} twice and using \eqref{claimeq} we find that for any fixed $u$ and $|t| \leq \tm$, 
\begin{equation}\label{relt0}
\left|\log \Esp_{\PNbeta}\left[ \exp\left(u t^2 \int \ep(x,y) d\fluct_N(x)d\fluct_N(y) \right)\right]\right|\le C_u t^2 N\|\psi\|_{C^3}^2.
\end{equation}

Combining \eqref{Laplace11}, \eqref{expan}, \eqref{vT2}, \eqref{t1t}, \eqref{T0}, we obtain that 
\begin{multline}
 \exp\(s  N\int \xi d\muV\)  \Esp_{\PNbetaV}\left[ \exp\left( s \Fluct_N(\xi) \right) \right] \\
= \exp \left(-\frac{\beta}{2} \left( N^2 t\int \xi d\muV + \hal N^2 t^2 \int \xi'\psi d\muV \)+ tN\left(1-\frac{\beta}{2}\right) \int \psi' d\muV \right)
\\
\Esp_{\PNbetaV} \( \exp\(t \frac{\beta}{2} \Ani[\XN,\psi] + \left(1-\frac{\beta}{2} \right) \int \log \phi_t' d\fluct_N +\Err\right)\right),
\end{multline}
with 
\begin{multline}
\label{def:ErrorLaplaceA}
\Err = - \frac{\beta}{2} \Fluct_N[2N \tau_t] - \frac{\beta}{2} t^2 \int \ep(x,y) d\fluct_N(x) d\fluct_N(y) \\ 
+ O(t^3 N^2 \|\psi\|_{C^1}^3 \comT{+ t^3N^2 \|\psi\|_{L^{\infty}}^2\|\xi\|_{C^2}} + t^2 N \|\psi\|_{C^1}^2).
\end{multline}
Combining the estimates \eqref{relT1}, \eqref{relt0} and using the Cauchy-Schwarz inequality, we see 
\begin{multline*}
\left|\log \Esp_{\PNbetaV} (\exp(u \Err))\right| \\
\le C_u \left(t^2 N \sqrt N  \|\psi\|_{C^2}^2 + t^2 N \|\psi\|_{C^3}^2 \comT{+ t^3N^2 \|\psi\|_{L^{\infty}}^2\|\xi\|_{C^2}} + t^3 N^2 \|\psi\|_{C^1}^3+ t^2 N \|\psi\|_{C^1}^2\right), 
\end{multline*}
which we may simplify as 
$$
\left|\log \Esp_{\PNbetaV} (\exp(u \Err))\right| \le C_u \left(t^2 N \sqrt{N} \|\psi\|_{C^3}^2 \comT{+ t^3N^2 \|\psi\|_{L^{\infty}}^2\|\xi\|_{C^2}} + t^3 N^2 \| \psi \|_{C^1}^3\right).
$$
\end{proof}

The following lemma shows that we can treat $\int \log \phi_t' d\fluct_N$ in the right-hand side of \eqref{Laplace} as an error term.
\begin{lem} \label{lem:logphifluctaserror}
For any fixed $u$, we have
\begin{equation}
\label{qterr}
\left|\log \Esp_{\PNbetaV} \left[\exp\left(u \left(1-\frac{\beta}{2}\right) \int \log \phi_t' d\fluct_N \right)\right]\right|\\
\le C_u  t\sqrt N \|\psi\|_{C^2}.
\end{equation}
\end{lem}
\begin{proof}
It follows from applying Corollary \ref{cor:concentration} to the map $\log \phi_t'$.
\end{proof}

Next, we deal with the term $\Ani[\XN,\psi]$ in \eqref{Laplace}.

\subsection{First control on the anisotropy term}
With Proposition \ref{fonda} at hand, the only thing that remains to elucidate is the behavior of the exponential moments of $\Ani[\XN,\psi]$, which we call the {\it anisotropy}. In particular we will show that these are $o(1)$.

Using concentration bounds, more precisely applying Proposition \ref{fluctenergy} twice together with \eqref{claimeq}, we  obtain a first bound
\begin{lem}\label{lemfacile}
For $|t| \leq \tm$ we have
\begin{equation}\label{resultani}
\left|\log \Esp_{\PNbetaV} (\exp(-\beta t \Ani[\XN,\psi]))\right| \leq C t N \|\psi\|_{C^3(U)}.
\end{equation}
\end{lem}
\begin{proof} Let us write
\begin{equation}\label{anir1}
\Ani[\XN, \psi]= \int g(x) d\fluct_N(x),
\end{equation}
where we let
\begin{equation}
g(x):= \int \hat \psi (x,y) d\fluct_N (y), \quad \hat\psi(x,y):= \frac{\psi(x)-\psi(y)}{x-y}.
\end{equation}
It is clear that 
 \begin{equation} \label{hatpsivspsi}
\|\hat \psi\|_{C^2(U \times U)}\le \|\psi\|_{C^3(U)}.
\end{equation}
Using Proposition \ref{fluctenergy} twice, 
we can thus write 
$$\|\nab g\|_{L^\infty}\le   \left|\int \nab_x \hat \psi(x,y) d\fluct_N (y)\right|\le C\|\nab_x\nab_y \hat \psi\|_{L^\infty} \(F_N(\XN, \muV)+ N \log N+ C N\)^{\hal}$$ and 
\begin{multline*}
|\Ani [\XN,\psi]|=\left|\int g(x)d \fluct_N(x)\right|\le C \|\nab g\|_{L^\infty}  \(F_N(\XN, \muV)+ N \log N+ C N\)^{\hal}\\ \le
C \|\hat \psi\|_{C^2(U\times U)}   \(F_N(\XN, \muV)+ N \log N+ C N\).\end{multline*}
In view of \eqref{claimeq} and \eqref{hatpsivspsi}, we deduce that
\begin{equation}\label{moyani}
\left|\log\Esp_{\PNbeta} \left[-\beta t\Ani [\XN,\psi]\right]\right|\le  Ct N \|\psi\|_{C^3(U)}.	
\end{equation}
\end{proof}

This shows that the exponential moments of the anisotropy yield bounded terms.

\subsection{Intermediary conclusion on the Laplace transform} 
 Inserting into the results of Proposition \ref{fonda} and Lemma \ref{lem:logphifluctaserror} we obtain the following (with $t = \frac{-2s}{\beta N}$)  \begin{multline*}
\left|\log \Esp_{\PNbetaV}\left[ \exp(s\Fluct_N(\xi)\right] \right| \le C\left( s \left( \|\psi\|_{C^1} + \|\psi\|_{C^3}\right) + s^2 \| \psi \|_{L^{\infty}} \|\xi\|_{C^1}\right)
\\
+ C\left(\frac{s}{\sqrt{N}} \|\psi\|_{C^3}^2 + \frac{s^3}{N} \|\psi\|_{L^{\infty}}^2\|\xi\|_{C^2} + \frac{s^3}{N} \| \psi \|_{C^1}^3 +  \frac{s}{\sqrt N} \|\psi\|_{C^2}  \right).
\end{multline*}

In view of \eqref{contpsixi}, we can bound $\| \psi \|_{C^{n}} \leq \| \xi \|_{C^{2\kk + 1 + n}}$ for any $n$, hence we obtain
 \begin{multline*}
\left|\log \Esp_{\PNbetaV}\left[ \exp\left(s\Fluct_N(\xi)\right)\right] \right|\\ \le C\left( s \left( \|\xi\|_{C^{2\kk+2}} + \|\xi \|_{C^{2\kk+4}}\right) + s^2 \| \xi \|_{C^{2\kk+1}} \|\xi\|_{C^1}\right)
\\
+ C \left(\frac{s}{\sqrt{N}} \|\xi\|_{C^{2\kk+4}}^2 + \frac{s^3}{N} \|\xi\|_{C^{2\kk+1}}^2 \|\xi\|_{C^2} + \frac{s^3}{N} \| \xi \|_{C^{2\kk+2}}^3 +  \frac{s}{\sqrt N} \|\xi \|_{C^{2\kk+3}}  \right).
\end{multline*}  
\comT{We may re-write the right-hand side as a less sharp but simpler bound. 
\begin{coro}\label{bornme}
Under the assumptions of Theorem \ref{theoreme} we have for any $s$ such that $\frac{2|s|}{\beta N} \leq \tm$
\begin{equation}\label{efluct}
\left|\log \Esp_{\PNbetaV} \left[ \exp(s\Fluct_N(\xi)) \right] \right| \leq C (s+s^3) \left(\|\xi\|^3_{C^{2\kk+4}} + \|\xi\|_{C^{2\kk+4}} \right)
\end{equation} 
where $C$ depends only on $\beta$ and $V$.
\end{coro}
}

\comT{The estimate \eqref{efluct} shows that fluctuations of a smooth enough test function are typically of order $1$, which is an improvement on the a priori bound \eqref{concentration} but does not yield a CLT. Let us observe that the only error term of order $1$ comes from \eqref{moyani}, which was derived by treating $\Ani[\XN, \psi]$ as a fluctuation and using the a priori bound.}

\comT{In the one-cut, non-critical case, this argument can be bootstrapped, as described in Appendix \ref{appb}: roughly speaking we use the new control \eqref{efluct} instead of \eqref{concentration} to estimate the exponential moments of $\Ani[\XN, \psi]$, and improve \eqref{moyani} by a factor $N$. The contribution of $\Ani[\XN, \psi]$ in \eqref{Laplace} becomes of lower order and Proposition \ref{fonda} yields the desired convergence of Laplace transforms. This is a standard technique, see e.g. the recursion of \cite{BorGui1}, and can be implemented in the one-cut, non-critical case because the operator $\Xi_V$ is invertible. In the multi-cut or critical cases, however, we only know how to invert the operator $\Xi_V$ under the extra conditions on the test function.}

We then use a different way to show that the exponential moments of $\Ani$ are in fact smaller than \eqref{moyani}, by leveraging on the expansion of $\log \ZNbetaV$ of \cite{ls1} quoted in Lemma \ref{lem:comparaison}. Indeed, comparing  \eqref{Laplace0} to \eqref{Laplace}, we observe that the expansion of $ \log \ZNbetaVt-\log \ZNbetaV$ provides another way of evaluating the exponential moments of $\Ani$.
More precisely, we will use the expansion of $\log \KNbeta(\tmut, \tzetat)-\log \KNbeta(\muV, \zetaV)$ where $\tmut$ is the approximate equilibrium measure obtained by pushing forward $\muV$ by $\id + t\psi$.

\section{Smallness of the anisotropy term and proof of Theorem \ref{theoreme}}
\subsection{Comparison of partition functions by transport}
\begin{defi}  For $t \in [-\tm,\tm]$, where $\tm$ is as in \eqref{def:tmax} we let $\PNbetat$ be the probability measure
\begin{equation} \label{def:PNbetat}
d\PNbetat(\XN) = \frac{1}{\KNbeta(\tmut, \tzetat)} \exp \left( - \frac{\beta}{2} \left(F_N(\XN, \tmut) + 2N \sum_{i=1}^N\tzetat(x_i) \right)\right) d\XN,
\end{equation}
where $\KNbeta(\tmut, \tzetat)$ is as in \eqref{def:KNbeta}.
\end{defi}

 \begin{lem}[Comparison of energies]
 \label{44}
Assume $\psi\in C^3(\R)$.  For any $\XN\in U^N$, letting
 $\Phi_t(\XN) = (\phi_t(x_1),\cdots, \phi_t(x_N))$,
  we have
 \begin{multline} \label{tmuttotmueqz}
\left| F_N(\Phi_t(\XN), \tmut)-F_N(\XN, \muV) - \sum_{i=1}^N \log\phi_t' (x_i) + \frac{t}{2} \, \Ani[\XN,\psi]\right|
\\ \le  C t^2 \left( F_N (\XN, \muV)+N\log N\right)\|\psi\|_{C^3}^2.
 \end{multline}
 \end{lem}
 \begin{proof}
 Since by definition $\tilde \mueqt=\phi_t\# \mu_0$ we may write 
  \begin{multline*}
 F_N(\Phi_t(\vec{X}_N), \tilde\mueqt)-F_N(\vec{X}_N, \muV) \\ = 
- \iint_{\R \times \R \backslash \triangle} \log |x-y|\Big(\sum_{i=1}^N \delta_{\phi_t (x_i)}- N \tilde\mueqt\Big)(x) \Big(\sum_{i=1}^N \delta_{\phi_t (x_i)}- N \tilde\mueqt\Big)(y)   \\
 +   \iint_{\R \times \R \backslash \triangle} \log |x-y| d\fluct_N(x) d\fluct_N(y) \\ 
 = -  \iint_{\R \times \R \backslash \triangle} \log \frac{|\phi_t(x) - \phi_t(y)|}{|x-y|} d\fluct_N(x) d\fluct_N(y) \\
 = -  \iint_{\R \times \R} \log \frac{|\phi_t(x) - \phi_t(y)|}{|x-y|} d\fluct_N(x) d\fluct_N(y) + \sum_{i=1}^N \log \phi_t'(x_i).
\end{multline*}
We may then recognize the term $T_0$ in \eqref{415} and use \eqref{T0} and Proposition \ref{fluctenergy} to conclude.  \end{proof}

 \begin{lem}[Comparison of partition functions]
 We have, for any $t$ small enough
 \begin{multline}\label{420}
 \frac{\KNbeta(\tmut, \tzetat) }{ K_{N,\beta}(\muV, \zetaV)} =  \exp\left(N \left(1-\frac{\beta}{2}\right) 
   \left( \Ent( \muV)-\Ent(\tilde \mueqt)\right) \right) \\ 
  \Esp_{\PNbeta^{(0)}}   \left( \exp\left(\frac{\beta}{2} t \Ani[\XN, \psi] + \Err_1(\XN) +  \Err_2   (\XN)  \right)  \right), 
\end{multline}
with error terms bounded by
\begin{align}
\label{Err1A} & |\log \Esp_{\PNbeta^{(0)}} [\exp(-2\Err_1(\XN))] |\le  C t^2 N\|\psi\|_{C^3}^2, \\
\label{Err2A} & |\log \Esp_{\PNbeta^{(0)} }[\exp(-2\Err_2(\XN))]|\le Ct \sqrt{N} \|\psi\|_{C^2}^2 . \end{align}
 \end{lem}
\begin{proof}
Starting from \eqref{def:KNbeta}, by  a change of variables and in view  of \eqref{tmuttotmueqz},  we may write 
\begin{equation}\label{chv}
\begin{split}
\KNbeta(\tmut, \tzetat) & = \int \exp \left( - \frac{\beta}{2} \Big(F_N(\Phi_t(\XN), \tmut) + 2N \sum_{i=1}^N\tilde{\zeta}_t\circ \phi_t(x_i)   \Big) + \sum_{i=1}^N \log \phi'_t(x_i) \right) \ d \XN  \\ 
& = \int \exp \left( - \frac{\beta}{2} \left(F_N(\Phi_t(\XN), \tmut) + 2N \sum_{i=1}^N{\zetaV}(x_i)   \right) + \sum_{i=1}^N \log \phi'_t(x_i)  \right) \ d \XN, 
\end{split}
\end{equation} 
since  ${\zetaV}= \tilde{\zeta}_t\circ \phi_t$ by definition. Using  Lemma \ref{44}  we may write 
\begin{multline}\label{419}
\frac{\KNbeta(\tmut, \tzetat)}{ K_{N,\beta}(\muV, \zetaV)}  =  \frac{1}{K_{N,\beta}(\muV, \zetaV)}  \int_{\R^N} \exp\left( -\frac{\beta}{2} \left(F_N(\XN, \mueqz) + 2 N \sum_{i=1}^N{\zeta}(x_i)\right)  \right. \\ 
\left. + \left(1-\frac{\beta}{2}\right) \sum_{i=1}^N \log \phi_t'(x_i) +\frac{\beta}{2} t \,\Ani + \Err_1(\XN)  \right) d\XN\\
= \Esp_{\PNbeta^{(0)}} \left(\exp\left(\left(1-\frac{\beta}{2}\right) \sum_{i=1}^N \log \phi_t'(x_i)  +\frac{\beta}{2} t \,\Ani + \Err_1(\XN) \right)  \right), 
\end{multline}
where the $\Err_1$ term is bounded as in \eqref{Err1A}. We may finally write $$\sum_{i=1}^N \log \phi_t'(x_i) = N \int_{\R}\log  \phi_t' \, d\muV+  \Err_2 (\XN) $$ with an $\Err_2$ term as in \eqref{Err2A}, since this term is the same as the one arising in \eqref{T1}. Finally, since by definition $\phi_t\# \muV = \tilde \mueqt$ we may observe that  $\phi_t'= \frac{    \muV}{ \tilde \mueqt  \circ \phi_t}$ and thus
\begin{equation}\label{logent}
 \int_{\R} \log \phi_t' \, d\muV =  \int_{\R} \log \muV \, d\muV - \int_{\R} \log \mueqt\circ \phi_t  \, d\muV = \Ent(\muV)-  \Ent(\tilde \mueqt).
 \end{equation}
This yields \eqref{420}.
\end{proof}

\subsection{Smallness of the anisotropy term}
\begin{prop} For any $s$ such that $\frac{2|s|}{\beta N} \leq \tm$, we have
\begin{equation}\label{anismallagain}
\log \Esp_{\PNbetaV} \left( \exp\left(\frac{-s}{N}\Ani\right) \right) = o_N(1).
\end{equation}
\end{prop}
\begin{proof}
Applying Cauchy-Schwarz to \eqref{420} we may write 
\begin{multline}\label{esp3}
\Esp_{\PNbetaV}\left[ \exp \left(\frac{\beta}{4} t\Ani\right)  \right]^2\\ \le \Esp_{\PNbetaV}\left[ \exp\left(\frac{\beta}{2} t \Ani +   \Err_1+\Err_2 \right)   \right] \Esp_{\PNbetaV}\left[ \exp ( - \Err_1-\Err_2)\right]\\
\le \frac{\KNbeta(\tmut, \tzetat)}{ K_{N,\beta}(\muV, \zetaV)} \exp\left( \left(1-\frac{\beta}{2}\right) N \left( \Ent(\tmut)-\Ent(\muV)\right) \right)    \\\Esp_{\PNbetaV}\left( \exp (- 2\Err_1) \) \Esp_{\PNbetaV} \( \exp(-2\Err_2 )\right). 
 \end{multline}  
Inserting  \eqref{expzformula} and \eqref{Err1A}--\eqref{Err2A} into \eqref{esp3} we obtain that for $t$ small enough,
\begin{equation}\label{anismall}
\log \Esp_{\PNbetaV}\left( \exp \left(\frac{\beta}{4} t\Ani\right) \right)\le   C ( t^2 N \|\psi\|_{C^3}^2 +t \sqrt N\|\psi\|_{C^2}^2)+N \delta_N,
\end{equation}
for some sequence $\{\delta_N\}_N$ with $\lim_{N \ti} \delta_N = 0$.
Applying this to $t=4\ep/\beta   $ with $\ep$  small \comT{(possibly depending on $N$)} and using H\"older's inequality, we deduce 
$$
\log \Esp_{\PNbetaV} \left( \exp\left(\frac{-s}{N}\Ani\right) \right) \le \frac{|s|}{N\ep} \log \Esp_{\PNbetaV} \left( \exp(\ep \Ani)  \right) \le C |s| \ep \|\psi\|_{C^3}^2+C \frac{|s|}{\sqrt N}\|\psi\|_{C^2}^2+ C \frac{|s|}{\ep}  \delta_N.
$$
In particular, choosing $\epsilon = \sqrt{\delta_N}$, we get \eqref{anismallagain}. 
\end{proof}

\subsection{Conclusion: proof of Theorem \ref{theoreme}}
\begin{proof}[Proof of Theorem \ref{theoreme}]
Combining \eqref{Laplace} and \eqref{qterr} for $t= -\frac{2s}{\beta N}$ (where $s$ is independent of $N$) and \eqref{anismallagain}, together with the Cauchy-Schwarz inequality, we find
\begin{multline}\label{421}
\log \Esp_{\PNbetaV} \left[\exp(s\Fluct_N(\xi))\right]= \log \Const  + o_N(1) 
\\
+ \comT{O\left(\frac{s^2+s}{\sqrt{N}} \left( \|\psi\|_{C^3(U)}^2  + \|\psi\|_{C^3(U)}^2 \right) + \frac{s^3}{N}\left( \|\psi\|_{C^1}^3 + \|\psi\|_{L^{\infty}}^2 \|\xi\|_{C^2} \right) \right)}
\end{multline}
with
\begin{equation}
\log \Const = -\frac{s^2}{\beta}\int \xi' \psi d\muV -\frac{2s}{\beta}\left(1-\frac{\beta}{2}\right) \int \psi' d\muV
\end{equation}

Letting $N\to \infty$, we obtain, 
\begin{equation}
\label{concl1}
\lim_{N\to \infty} \log \Esp_{\PNbetaV} [\exp(s\Fluct_N(\xi))]=-\frac{s^2}{\beta}\int \xi' \psi d\muV - \frac{2s}{\beta}\left(1-\frac{\beta}{2}\right)\int \psi' d\muV 
\end{equation}
and the rate of convergence is uniform for $s$ in a compact set of $\R$.

Thus the Laplace transform of $\Fluct_N(\xi)$ converges (uniformly on compact sets) to that of a Gaussian of mean $m_\xi$ and variance $v_\xi$, which implies convergence in law and proves the main theorem.
 \end{proof}

\appendix

\section{The one-cut regular case} \label{appb}
In the one-cut noncritical case, every regular enough function  is in the range of the operator $\Xi$, so that the map $\psi$ can always be built.  This allows to bootstrap the approach used for proving Theorem \ref{theoreme}.
In this appendix, we expand on how we can proceed in this simpler setting without refering to the result of \cite{ls1} but assuming more regularity of $\xi$, and retrieve the findings of \cite{BorGui1}  (but without assuming analyticity), as well as a rate of convergence  for the Laplace transform of the fluctuations. 

 \subsection{The bootstrap argument}
 We will consider the whole family  $\PNbetat$ of probability measures
\begin{equation*}
d\PNbetat(\XN) = \frac{1}{\KNbeta(\tmut, \tzetat)} \exp \left( - \frac{\beta}{2} \left(F_N(\XN, \tmut) + 2N \sum_{i=1}^N\tzetat(x_i) \right)\right) d\XN,
\end{equation*}
where $\KNbeta(\tmut, \tzetat)$ is as in \eqref{def:KNbeta}. We will also emphasize the $t$ dependence by writing 
$$
\fluct_N^{(t)} := \sum_i \delta_{x_i}- N \tmut
$$
and using similarly the notation $\Fluct^{(t)}$ and $\Ani^{(t)}$.
 
 Let us first explain the  main computational  point for the bootstrap argument.
 Differentiating \eqref{Laplace} with respect to $t$ and using \eqref{formeerr}, we obtain 
 \begin{equation*}
 \frac{-\beta N}{2}\Esp_{\PNbeta^{(0)} } [\Fluct_N^{(0)} (\xi)] 
 =  \Esp_{\PNbeta^{(0)}} \left[-\frac{\beta}{2} \Ani^{(0)}[\XN,\psi] + \left(1-\frac{\beta}{2} \right) \frac{d}{dt}_{|t=0}\sum_{i=1}^N \log \phi_t'(x_i) \right].\end{equation*}
 Note that here all the error terms in \eqref{Laplace} have disappeared because they were in factor of $t^2$. Also
 this is true as well for all $t \in [-\tm, \tm]$, i.e. 
\begin{equation}\label{bootstrap0}
\Esp_{\PNbeta^{(t)} } [\Fluct_N^{(t)} (\xi)] 
 = - \frac{2}{\beta N} \Esp_{\PNbeta^{(t)}} \left[-\frac{\beta}{2} \Ani^{(t)}[\XN,\psi] + \left(1-\frac{\beta}{2}\right) \frac{d}{dt}\sum_{i=1}^N \log \phi_t'(x_i) \right].
 \end{equation}
  We may in addition write that 
 \begin{equation}
 \label{logp}\frac{d}{dt}\sum_{i=1}^N \log \phi_t'(x_i) = N\int \frac{d}{dt} \log \phi_t' \,d \tmut+ \Fluct_N^{(t)} \left( \frac{d}{dt} \log \phi_t' \right)
 \end{equation} 
so that  
 \begin{multline}\label{bootstrap}
\Esp_{\PNbeta^{(t)} } [\Fluct_N^{(t)} (\xi)] 
 =- \frac{2}{\beta } \left( 1- \frac{\beta}{2}\right)\int \frac{d}{dt} \log  \phi_t' \,d \tmut
 \\ -  \frac{2}{\beta N} \Esp_{\PNbeta^{(t)}} \left[-\frac{\beta}{2} \Ani^{(t)}[\XN,\psi] + \left(1-\frac{\beta}{2}\right) \Fluct_N^{(t)} \(\frac{d}{dt} \log \phi_t'\) \right].
 \end{multline}
 This provides a  functional equation  which gives the expectation of the fluctuation  in terms of a constant term plus a lower order expectation of another fluctuation and the $\Ani$ term (which itself can be written as a fluctuation, as noted below), allowing to expand it in powers of $1/N$ recursively.

% Let us note  that by definition   $\phi_t\# \mueqz = \tilde \mueqt$ we may observe that  $\phi_t'= \frac{\tilde \mueqt}{\mueqz \circ \phi_t}$ and thus
%\begin{equation}\label{logent}
% \int_{\R}  \log \phi_t' \, d\mueq = \Ent(\tilde \mueqt)-\Ent(\mueq),
% \end{equation}which will allow to handle the constant term involving $\int \frac{d}{dt} \log \phi_t' \,d \tmut$ after integration in $t$.

\subsection{Improved control on the fluctuations}
%Let us now be more specific and start with the following lemma which is valid in the general situation.

Assuming from now on that $\nn=0$ and $\mm=0$ so that every regular function is in the range of $\Xi_V$, since $\tmut$ is the push forward of $\muV$ by a regular map, it  is also one-cut, thus all the results proved thus far remain true for $\PNbeta^{(t)}$ and for any   regular enough test function $\xi$. 
 Thanks to this,  we can upgrade the control of exponential moments given in Corollary \ref{bornme} into the control of a weak norm of $\Fluct_N^{(t)}$.  Here we use the Sobolev spaces $H^\alpha(\R)$.

\begin{lem}\label{lems2}
Under the same assumptions, for $\alpha\ge \comT{14}$  we have
\begin{equation}
\left|  \Esp_{\PNbeta^{(t)}} \left[ \|\fluct_N^{(t)}\|_{H^{-\alpha}}^2\right]\right|\le C,
\end{equation} where $C$ depends only on $V$.
\end{lem}

\begin{proof}
The proof is inspired by \cite{Armstrong:2017qy}, in particular we start from \cite[Prop. D.1]{Armstrong:2017qy} which states that 
\begin{equation}
\|u\|_{H^{-\alpha}(\R)}^2 \le C \int_0^1  r^{ \alpha  -1} \|u* \Phi(r, \cdot) \|_{L^2(\R)}^2 \, dr
\end{equation}
where $\Phi( r, \cdot)$ is the standard heat kernel, i.e. $\Phi(r,x)= \frac{1}{\sqrt{4\pi r}} e^{-\frac{|x|^2}{4r}} .$
It follows that 
\begin{equation}\label{ef1}
\Esp_{\PNbeta^{(t)}}  \left[\|\fluct_N^{(t)} \|_{H^{-\alpha}(\R)}^2 \|\right]  \le C \int_0^1 r^{\alpha-1} \Esp_{\PNbeta^{(t)}  } \left[\|\fluct_N^{(t)} * \Phi(r, \cdot)\|_{L^2(\R)}^2\right]  \, dr .\end{equation}
On the other hand we may easily check that, letting $\xi_{x,r}:=  \Phi(r,x-\cdot)$, we have
 \begin{equation}\label{ef2}
 \Esp_{\PNbeta^{(t)}} \left[\|\fluct_N^{(t)}*\Phi(r, \cdot) \|^2_{L^2(\R)}\right] = \int \Esp_{\PNbeta^{(t)}} \left[ \(\Fluct_N^{(t)} (\xi_{x,r})  \)^2\right] \, dx .
 \end{equation}
Applying the result of  Corollary \ref{bornme}  to $\xi_{x,r}$ gives us  a control on the second moment of $\Fluct_N^{(t)}[\xi_{x,r}]$ \comT{of the form 
$$ 
\Esp_{\PNbeta^{(t)}} \left[(\Fluct_N^{(t)}(\xi_{x,r}))^2\right]  \le C\left( 
\|\xi_{x,r}\|_{C^4}^3 + \|\xi_{x,r}\|_{C^4} \right).
$$
Inserting into \eqref{ef1} and \eqref{ef2}, we are led to 
\begin{equation*}
\Esp_{\PNbeta^{(t)}}  \left[\|\fluct_N^{(t)} \|_{H^{-\alpha}(\R)}^2 \right] 
 \le C\int_0^1 \int r^{\alpha-1}  C\left( 
\|\xi_{x,r}\|_{C^4}^3 + \|\xi_{x,r}\|_{C^4} \right) dx \, dr.
\end{equation*}
}

Since $U$ is bounded, the right-hand side can be bounded by $C \int_0^1 r^{\alpha-1} (1 + r^{-27/2})\, dr$, which converges if \comT{$ \alpha > 27 /2$}.
\end{proof}

\subsection{Proof of Theorem \ref{th2}}
 First, by  \eqref{chv} and 
in view of  Lemma \ref{44}, we may  write 
\begin{equation}
\frac{d}{dt}_{|t=0} \log K_{N,\beta}(\tmut,\tilde{\zeta_t})= \Esp_{\PNbeta^{(0)}} \left[-\frac{\beta}{2} \Ani^{(0)}[\XN,\psi] + \left(1-\frac{\beta}{2}\right) \frac{d}{dt}_{|t=0}\sum_{i=1}^N \log \phi_t'(x_i) \right].\end{equation}
Similarly, we have  for all $t$
\begin{equation}\label{dtlk}\frac{d}{dt} \log K_{N,\beta}(\tmut,\tilde{\zeta_t})= \Esp_{\PNbeta^{(t)}} \left[-\frac{\beta}{2} \Ani^{(t)}[\XN,\psi] + \left(1-\frac{\beta}{2}\right) \frac{d}{dt} \sum_{i=1}^N \log \phi_t'(x_i) \right].
\end{equation}
Indeed, $\tmut $ has the same regularity as $\muV$.

For any test function $\phi(x,y)$  we may  write 
$$
\iint \phi(x,y) d\fluct_N^{(t)}(x)\, d\fluct_N^{(t)}(y) \le \|\phi\|_{C^{2\alpha}(U \times U)} \|\fluct_N^{(t)}\|_{H^{-\alpha}(\R)}^2 $$
and so by the result of  Lemma \ref{lems2}, we find 
\begin{equation} \label{dfluctdfluct}
\left|\Esp_{\PNbeta^{(t)}}\left( \iint \phi(x,y) d\fluct^{(t)}_N(x)\, d\fluct^{(t)}_N(y) \right)\right| \le C \|\phi\|_{C^{2\alpha}(U\times U)}.
\end{equation}
We may return to \eqref{anir1} and, using \eqref{dfluctdfluct}, write that 
\begin{equation}\label{anir2}
\left|\Esp_{\PNbeta^{(t)} }\left[\Ani^{(t)} [\XN,\psi]\right] \right|\le   C \|\psi\|_{C^{2\alpha+1}(U)}.
\end{equation}
On the other hand, by  differentiating \eqref{efluct} applied with $\xi= \frac{d}{dt} \log \phi_t'$,  we have 
\begin{equation}\label{anir3}
\left|\Esp_{\PNbeta^{(t)}}\left[ \int \frac{d}{dt}\log \phi_t'     d\fluct_N^{(t)}\right]   \right| \le \comT{C\( \|\psi\|_{C^{5}(U)}+ \|\psi\|^3_{C^{5}(U)} \)}
\end{equation}

Inserting  \eqref{logent} and \eqref{anir2} and \eqref{anir3},  \eqref{logp} into \eqref{dtlk}, and integrating between $0$ and  $t=-2s/N\beta$,
we obtain 
\begin{equation}
\label{partifuncABCD}
\log\frac{ K_{N,\beta}(\tmut,\tilde{\zeta_t}) }{ \KNbeta(\muV,\zetaV)}=\(1- \frac{\beta}{2}\)  N \( \Ent(\tilde \mueqt)-\Ent(\mu_0)\)   +\frac{s}{N} 
O\comT{C\( \|\psi\|_{C^{5}(U)}+ \|\psi\|^3_{C^{5}(U)} \)}.
\end{equation}
Comparing \eqref{partifuncABCD} with \eqref{420}, we obtain 
\begin{equation*}
\log \Esp_{\PNbeta^{(0)}} \left( \exp \left(\frac{\beta}{2}t \Ani^{(0)} +  \Err_1(\XN)+\Err_2(\XN) \right)\right) = \frac{s}{N} \comT{O\left( \|\psi\|_{C^{2\alpha+1}(U)} + \|\psi\|_{C^{5}(U)} + \|\psi\|^3_{C^{5}(U)} \right)}
\end{equation*}

Using the bounds of \eqref{Err1A}-\eqref{Err2A} and the Cauchy-Schwarz inequality, we deduce that 
$$
\log  \Esp_{\PNbeta^{(0)}} \left[ \exp \left(\frac{\beta}{2}t \Ani^{(0)} \right) \right]
= \frac{s}{N} O\left( \|\psi\|_{C^{2\alpha+1}(U)} + \|\psi\|_{C^{5}(U)} + \|\psi\|^3_{C^{5}(U)}  + s \|\psi\|_{C^3}^2 + \sqrt{N} \| \psi \|_{C^2} \right).$$

 This can be inserted in place of \eqref{anismallagain}  into  \eqref{Laplace} yields 
\begin{multline*}
\left|\log \Esp_{\PNbetaV}\left[ \exp(s\Fluct_N(\xi))\right] + \left(1-\frac{\beta}{2}\right) \frac{2s}{\beta} \int \psi' d\muV+ \frac{s^2}{\beta} \int \xi' \psi d\muV \right|\\ =  \frac{s}{N} O\left( \|\psi\|_{C^{2\alpha+1}(U)} + \|\psi\|_{C^{5}(U)} + \|\psi\|^3_{C^{5}(U)}  + s \|\psi\|_{C^3}^2 + \sqrt{N} \| \psi \|_{C^2} \right).
\end{multline*}
Taking $\alpha = 14$, this proves Theorem \ref{th2}.

\subsection{Iteration and expansion of the partition function to arbitrary order}
Let $V, W$ be two $C^{\infty}$ potentials, such that the associated equilibrium measures $\mu_V, \mu_W$ satisfy our assumptions with $\nn=0, \mm=0$. In this section, we explain how to iterate the procedure described above to obtain a relative expansion of the partition function, namely an expansion of $\log \ZNbeta^{W} - \log \ZNbeta^{V}$ to any order of $1/N$. Up to applying an affine transformation to one of the gases, whose effect on the partition function is easy to compute, we may assume that $\mu_V$ and $\mu_W$ have the same support $\Sigma$, which is a line segment.

Since $V, W$ are $C^{\infty}$ and $\mu_V, \mu_W$ have the same support and a density of the same form \eqref{formmu} which is $C^{\infty}$ on the interior of $\Sigma$, the optimal transportation map (or monotone rearrangement) $\phi$ from $\mu_V$ to $\mu_W$ is $C^{\infty}$ on $\Sigma$ and can be extended as a $C^{\infty}$ function with compact support on $\R$.  We let $\psi := \phi - \id$, which is smooth, and for $t \in [0,1]$ the map $\phi_t := \id + t \psi$ is a $C^{\infty}$-diffeomorphism, by the properties of optimal transport. We let $\tmut := \phi_t \# \mu_V$ as before. 

We can integrate \eqref{dtlk} to obtain
\begin{multline*}
\log\frac{ \KNbeta(\mu_{W}, \zeta_{W} )  }{\KNbeta(\mueq, \zeta_V)}\\ = \int_0^1 \Esp_{\PNbeta^{(t)} } \left[-\frac{\beta}{2} \Ani^{(t)} [\XN,\psi] + \left(1-\frac{\beta}{2}\right) N \int \frac{d}{dt} \log \phi_t'\, d\tmut + \left(1-\frac{\beta}{2}\right) \int \frac{d}{dt}\log \phi_t' d\fluct_N^{(t)}\right]\, dt \\
= N\left(1-\frac{\beta}{2}\right) (\Ent(\mu_{W})- \Ent(\mu_V))  \\ + \int_0^1 \Esp_{\PNbeta^{(t)} } \left[-\frac{\beta}{2} \Ani^{(t)} [\XN,\psi] +  \left(1-\frac{\beta}{2}\right) \Fluct_N\left[ \int \frac{d}{dt}\log \phi_t' d\fluct_N^{(t)}\right] \right]\, dt .
\end{multline*}
The integral on the right-hand side is of order 1, and we claim that the terms in the integral can actually be computed and expanded up to an error $O(1/N)$ using the previous lemma. 
This is clear for the term $\Esp_{\PNbeta^{(t)}} \left[ \Fluct_N^{(t)}(\frac{d}{dt}\log \phi_t' ) \right]$ which can be computed up to an error $O(1/N)$ by the result of Theorem \ref{th2}.
The term $\Esp_{\PNbeta^{(t)} } \left[-\frac{\beta}{2} \Ani^{(t)} [\XN,\psi]\right]$ can on the other hand be deduced from the knowledge of the covariance structure of the fluctuations. Let $\mathcal{F}$ denote the  Fourier transform.  
In view of  \eqref{anir1}, using the identity 
\begin{equation*}
\frac{\psi(x) - \psi(y)}{x-y} = \int_0^1 \psi'(sx + (1-s)y) ds
\end{equation*}
\noindent and the Fourier inversion formula we may write  
\begin{multline}
\Esp_{\PNbeta^{(t)}    }\left[\Ani^{(t)} [\XN,\psi]\right]= \Esp_{\PNbeta^{(t)}} \left[\iint_{\R \times \R}  \int_0^1 \psi'(sx + (1-s)y) ds  \, d\fluct_N^{(t)} (x) d\fluct_N^{(t)} (y) \right]\\
 = \int \int_0^1  \lambda \mathcal{F}(\psi)(\lambda) \Esp_{\PNbeta^{(t)}} \left[\Fluct_N^{(t)} (e^{is\lambda \cdot}) \Fluct_N^{(t)} (e^{i(1-s)\lambda \cdot}) \right] ds\, d\lambda .\end{multline} 
 On the other hand, let $\varphi_{s,\lambda}$ be  the  map associated to $e^{is\lambda \cdot}$ by Lemma \ref{leminv}.
Separating the real part and the imaginary part we may use the results of  the previous subsection to $e^{is \lambda \cdot}$ and obtain
$$
\Esp_{\PNbeta^{(t)}} \left[\Fluct_N^{(t)} (e^{is\lambda \cdot})\right]=  \( 1- \frac{2}{\beta}\)\int \varphi_{s,\lambda}'d\tmut + O(\frac{1}{N}) \ .
$$  
By polarization of the expression for the variance (see \eqref{variance2}) and  linearity
\begin{multline*}
\Esp_{\PNbeta^{(t)}} \left[\Fluct_N^{(t)} (e^{is\lambda \cdot})\Fluct_N^{(t)} (e^{i(1-s)\lambda \cdot}) \right]
= \Esp_{\PNbeta^{(t)}} \left[\Fluct_N^{(t)} (e^{is\lambda \cdot})\right] \Esp_{\PNbeta^{(t)}} \left[\Fluct_N^{(t)} (e^{i(1-s)\lambda \cdot})\right] \\+ \frac{2}{\beta}\Big(  \iint \(\frac{\varphi_{s,\lambda}(u)-\varphi_{s,\lambda} (v)}{u-v}\) \(\frac{\varphi_{(1-s), \lambda}(u)-\varphi_{(1-s), \lambda} (v)}{u-v}\) d\tmut(u) d\tmut(v) \\+\int V_t'' \varphi_{s,\lambda} \varphi_{(1-s), \lambda}    d\tmut \Big) + O(\frac{1}{N}).
\end{multline*}
  Letting $N\to \infty$,  we may then find the expansion up to $O(1/N)$ of  $\Esp_{\PNbeta^{(t)} } \left[-\frac{\beta}{2} \Ani^{(t)} [\XN,\psi]\right]$.
 Inserting it into the integral gives a relative expansion to order $1/N$ of the (logarithm of the) partition function $\log \KNbeta$. This procedure can then be iterated to yield a relative expansion to arbitrary order of $1/N$ as desired.

 \section{Auxiliary proofs}\label{auxi}
\subsection{Proof of Lemma \ref{lem:splitting}} \label{sec:preuvesplitting}
\begin{proof}
Denoting $\triangle $ the diagonal in $\R\times \R$ we may write 
\begin{multline*}
 \HNV(\XN)  =   \sum_{i \neq j} -\log|x_i- x_j| + N \sum_{i=1}^N V(x_i)\\
=\iint_{\triangle^c} -\log|x-y| \Big(\sum_{i=1}^N \delta_{x_i}\Big) (x)  \Big(\sum_{i=1}^N \delta_{x_i}\Big) (y) + N \int_{\R} V(x) \Big(\sum_{i=1}^N \delta_{x_i}\Big)(x).
\end{multline*}
Writing $\sum_{i=1}^N \delta_{x_i}$ as $N \muV + \fluct_N$ we get
\begin{multline}
\label{finh}
 \HNV(\XN) =   N^2 \iint_{\triangle^c} -\log|x-y| d\muV(x) d\muV(y) + N^2 \int_{\R} V d\muV
 \\ + 2N \iint_{\triangle^c} -\log|x-y| d\mueq(x) d\fluct_N(y) + N \int_{\R} V d\fluct_N 
 \\ + \iint_{\triangle^c} -\log|x-y| d\fluct_N(x)d\fluct_N(y).
\end{multline}
We now recall that $\zetaV$ was defined in \eqref{zeta}, and that $\zetaV =0$  in $\SigmaV$.
With the help of this we may rewrite the medium line in the right-hand side of \eqref{finh} as  
\begin{multline*}
2N \iint_{\triangle^c} -\log|x-y| d\muV(x) d\fluct_N(y) + N \int_{\R} V d\fluct_N \\
 = 2N  \int_{\R} \left(-\log|\cdot | * d\muV)(x) + \frac{V}{2}\right) d\fluct_N = 2N  \int_{\R} (\zetaV + c) d\fluct_N \\
 = 2N \int_{\R} \zetaV d    \Big(\sum_{i=1}^N \delta_{x_i}- N\muV\Big)= 2N \sum_{i=1}^N \zetaV(x_i).
\end{multline*}
The last equalities are due to the facts that  $\zetaV$ vanishes on the support of $\muV$ and that  $\fluct_N$ has a total mass $0$ since $\muV$ is a  probability measure.   We may also  notice that since  $\muV$ is absolutely continuous with respect to the Lebesgue measure, we may include the diagonal  back into the domain of integration.
By that same argument, one may recognize in the first line of the right-hand side of \eqref{finh} the quantity $N^2 \mathcal{I}_V(\muV)$. 
\end{proof} 
 
 \subsection{Proof of Proposition \ref{fluctenergy}} \label{sec:preuvefluctenergy}
We follow the energy approach introduced in \cite{SS15a,PS15}, which views the energy as a Coulomb interaction in the plane, after embedding the real line in the plane.
 We view $\R$ as identified with $\R\times \{0\} \subset \R^{2}= \{ (x,y), x\in \R, y\in \R\}$. 
 Let us denote by $\delta_{\R}$ the uniform measure on $\R\times \{0\}$, i.e. such that for any smooth $\varphi(x,y) $ (with $x\in \R, y \in \R$) 
we have 
$$\int_{\R^{2}} \varphi \delta_{\R}= \int_{\R  } \varphi(x,0) \, dx.$$

Given $(x_1, \dots, x_N)$ in $\R^N$, we identify them with the points $(x_1, 0), \dots, (x_N,0)$ in $\R^{2}$. For a fixed $\XN$ and a given probability density $\mu$ we introduce the electric potential $H_N^{\mu}$ 
by
\begin{equation}
H_N^\mu = (-\log |\cdot|) * \left(\sum_{i=1}^N \delta_{(x_i,0)} - N\mu \delta_{\R}\right).
\end{equation} 
%This is closely related to the Stieltjes transform.
Next, we define versions of this potential which are truncated hence regular near the point charges. For that let $\delta_{x}^{(\eta)}$ denote the uniform measure of mass $1$ on $\partial B(x,\eta)$ (where $B$ denotes an  Euclidean ball in $\R^2$). We  define  $H_{N,{\eta}}^{\mu}$  in $\R^{2}$ by 
\begin{equation} H_{N,{\eta}}^{\mu} = (-\log |\cdot|)  * \left(\sum_{i=1}^N \delta_{(x_i,0)}^{(\eta)} - N \mu \delta_{\R}\right).
\end{equation} 
These potentials make sense as functions in $\R^{2}$ and are harmonic outside of the real axis.  Moreover, $H_{N,\eta}^{\mu}$ solves 
\begin{equation}\label{dhne}
-\Delta H_{N,\eta}^{\mu}= 2\pi \left( \sum_{i=1}^N \delta_{(x_i,0)}^{(\eta)} - N \mu \delta_{\R}\right).
\end{equation}

\begin{lem} \label{FNasHN}
For any probability density $\mu$, $\XN$ in $\R^N$ and $\eta$ in $(0,1)$,  we have 
\begin{equation}
 F_N(\vec{X}_N, \mu)  \ge \frac{1}{2\pi}   \int_{\mathbb{R}^2} |\nabla H_{N,{\eta}}^{\mu} |^2  +  N \log \eta- 2N^2  \|\mu\|_{L^\infty}\eta
 .\end{equation} 
\end{lem}
\begin{proof}
First we notice that $\int_{\R^2}|\nab H_{N, {\eta}}|^2 $ is a convergent integral and that 
\begin{equation}\label{intdoub}\int_{\R^2}|\nab H_{N, \vec{\eta}}|^2=2\pi \iint -\log|x-y|d\left( \sum_{i=1}^N \delta_{x_i}^{(\eta)} -N\mu\drd\right)(x) d\left( \sum_{i=1}^N \delta_{x_i}^{(\eta)} -N\mu\drd\right)(y).\end{equation}
Indeed, we may choose $R$ large enough so that all the points of $\XN$ are contained in the ball $B_R = B(0, R)$.
    By Green's formula  and \eqref{dhne}, we have
\begin{equation}\label{greensplit1}
\int_{B_R} |\nabla H_{N,{\eta}}|^2 \\= \int_{\partial B_R} H_{N,{\eta}} \frac{\partial H_N}{\partial \nu} + 2\pi\int_{B_R} H_{N,{\eta}} \left(   \sum_{i=1}^N \delta_{x_i}^{(\eta)}-N\mu\drd\right)  .
\end{equation}
In view of the decay of $H_N$ and $\nab H_N$, the boundary integral tends to $0$ as $R \to \infty$, and so we may write 
$$\int_{\R^2}|\nab H_{N, {\eta}} |^2= 2\pi \int_{\R^2} H_{N,{\eta}} \left(   \sum_{i=1}^N \delta_{x_i}^{(\eta)}-N\mu\right)  $$ and thus \eqref{intdoub} holds.
We may next write 
\begin{multline}
\label{l3}\iint -\log|x-y|d\left( \sum_{i=1}^N \delta_{x_i}^{(\eta)} -N\mu\drd\right)(x) d\left( \sum_{i=1}^N \delta_{x_i}^{(\eta)} -N\mu\drd\right)(y)\\-\iint_{\triangle^c} -\log|x-y|\, d\fluct_N(x)\, d\fluct_N(y)\\
 =- \sum_{i=1}^N \log \eta  + \sum_{i\neq j} \iint -\log|x-y| \left(\delta_{x_i}^{(\eta)} \delta_{x_j}^{(\eta)} -   \delta_{x_i}\delta_{x_j}\right)+2 N\sum_{i=1}^N\iint -\log|x-y|\left( \delta_{x_i}-\delta_{x_i}^{(\eta)} \right) \mu.
 \end{multline}
\comT{We have used the fact that for any $x_i$, 
$$
\iint - \log |x-y| \delta_{x_i}^{(\eta)}(x) \delta_{x_i}^{(\eta)}(y) = - \log \eta,
$$ 
as follows from a direct computation of Newton's theorem.}
 
Let us now observe that $\int -\log|x-y|\delta_{x_i}^{(\eta)} (y)$, the potential generated by $\delta_{x_i}^{(\eta)}$ is equal to $\int -\log|x-y| \delta_{x_i}$ outside of $B(x_i,\eta)$, and smaller otherwise. Since its  Laplacian is $-2\pi \delta_{x_i}^{(\eta)}$, a negative measure,  this  is also a superharmonic function, so by the maximum principle, its value at a point $x_j$ is larger or equal to its average on a sphere centered at $x_j$. Moreover, outside $B(x_i,\eta)$ it is a harmonic function, so its values are equal to its averages. We deduce from these considerations, and reversing the roles of $i $ and $j$,  that for each $i\neq j$,
$$-\int\log|x-y| \delta_{x_i}^{(\eta)} \delta_{x_j}^{(\eta)} \le -\int\log|x-y| \delta_{x_i} \delta_{x_j}^{(\eta)}
\le -\int\log|x-y| \delta_{x_i} \delta_{x_j}.$$
We may also obviously  write  $$\int-\log|x-y| \delta_{x_i} \delta_{x_j} - \int - \log|x-y| \delta_{x_i}^{(\eta)} \delta_{x_j}^{(\eta)}
\le -\log|x_i-x_j| \indic_{ |x_i-x_j |\le 2\eta}.$$
We conclude that the second term  in the right-hand side of \eqref{l3} is nonpositive, equal to  $0$ if all the balls are disjoint, and bounded below by $\sum_{i\neq j} \log|x_i-x_j| \indic_{ |x_i-x_j |\le 2\eta}$.
Finally, by the above considerations,  since $\int -\log|x-y| \delta_{x_i}^{(\eta)}$ coincides with $\int -\log|x-y|\delta_{x_i}$ outside $B(x_i,\eta)$, 
we may rewrite the last term in the right-hand side of \eqref{l3} as 
$$2 N\sum_{i=1}^N\int_{B(x_i,\eta)} ( -\log|x-x_i|+ \log \eta)) d\mu\drd.$$ But we have that 
\begin{equation}\label{intfeta}\int_{B(0,\eta)} (  -\log |x|+\log \eta) \drd= \eta\end{equation} so if $\mu \in L^\infty$, this last term is bounded by $2\|\mu\|_{L^\infty} N^2\eta$.
Combining with all the above results yields the proof.
\end{proof}
\def \chitilde{\tilde{\chi}}
\begin{proof}[Proof of Proposition \ref{fluctenergy}]
We now apply Lemma \ref{FNasHN} for $\muV$ with $\eta=\frac{1}{2N}$. We obtain
\begin{equation}\label{premiereb}
  \frac{1}{2\pi}   \int_{\mathbb{R}^2} |\nabla H_{N,{\eta}}^{\mu} |^2 \le F_N(\vec{X}_N, \muV) + N\log N   + C(\|\muV\|_{L^\infty}+1) N .
 \end{equation}
Let $\xi$ be a Lipschitz, compactly supported test function in $\R$, and let $\chi(y)$ be a smooth cutoff function such that $\chi(y)=1$ for $|y|\le 1$, $\chi(y)= 0$ for $|y|\ge 2$ and $\|\chi' \|_{L^\infty}\le 1$. We then extend $\xi$ in $\R^2$ by $\chitilde$ defined as
$$
\chitilde(x,y) := \xi(x) \chi(y).
$$
It is easy to check that for any $(x,y)$,
$$|\nabla \chitilde(x,y)| \leq |\xi'(x)| + |\xi(x)|,$$
and $\chitilde$ is supported in an horizontal stripe of width $1$.

Letting $\#I $ denote the number of balls $B(x_i, \eta)$ intersecting the support of $\xi$, we have (with $\eta = \frac{1}{2N}$)
\begin{multline} \label{rel3ma}
\left|\int \left(\fluct_N- \left( \sum_{i=1}^N \delta_{x_i}^{(\eta)} - N\mueq\drd\right) \right) \chitilde \right|=  
\left|\int  \left( \sum_{i=1}^N( \delta_{x_i} - \delta_{x_i}^{(\eta)}) \right) \chitilde \right|
\\
\le \# I  \eta \|\xi'\|_{L^\infty} \leq \|\xi'\|_{L^\infty}, 
\end{multline}
where we have bounded $\# I$ by $N$ in the last inequality. 

In view of \eqref{dhne}, we also have 
\begin{multline}\label{rel3}
\left|\int  \left( \sum_{i=1}^N \delta_{x_i}^{(\eta)} - N\muV\drd\right) \chitilde \right|= \frac{1}{2\pi}\left| \int_{\R^2}\nabla H_{N,\eta}^{\muV} \cdot \nab( \chitilde) \right| \\
\le  \frac{1}{2\pi}(  \|\xi'\|_{L^2(\R)} + \|\xi\|_{L^2(\R)}) \|\nab H_{N,\eta}^{\muV}\|_{L^2(\R^2)}.
\end{multline}

Combining \eqref{premiereb}, \eqref{rel3ma} and \eqref{rel3}, we obtain
 \begin{multline}
 \left|\int   \xi \, \fluct_N\right| \\
 \le       \|\xi'\|_{L^\infty} +   \left(  \|\xi'\|_{L^2(\R)}+\|\xi\|_{L^2(\R)}\right)\left( F_N(\vec{X}_N, \muV) + N\log N   + C(\|\muV\|_{L^\infty}+1) N\right)^{\frac{1}{2}}.
 \end{multline}
\end{proof}

\subsection{Proof of Lemma \ref{lemmareg}}
\label{sec:preuvelemmareg}
\begin{proof}
Since $\muV$ minimizes the logarithmic potential energy \eqref{energy}, for any bounded continuous function $h$, \eqref{equil} holds. Of course, an identity like \eqref{equil} extends to complex-valued functions, and applying it to $h = \frac{1}{z - \cdot}$ for some fixed $z \in \mathbb{C} \setminus{\SigmaV}$ leads to 
\begin{equation} \label{identiteG}
G(z)^2 - G(z)V'(\Re(z)) + L(z)= 0,
\end{equation}
where $G$ is the usual Stieltjes transform of $\muV$
\begin{equation}\label{eqG}
G(z) = \int \frac{1}{z-y}d\muV(y),
\end{equation}
and $L$ is defined by
\begin{equation} \label{eqL}
 L(z) = \int  \frac{V'(\Re(z)) - V'(y)}{z - y}d\muV(y).
\end{equation}

Solving \eqref{identiteG} for $G$ yields 
\begin{equation}\label{solvG}
G(z) = \frac{1}{2} \left( V'(\Re(z)) - \sqrt{V'(\Re(z))^2 - 4  L(z)} \right).
\end{equation}
As is well-known, \comT{since $\muV$ is continuous on $\Sigma_V$, the quantity} $ -\frac{1}{\pi} \Im(G(x + i \epsilon))$ converges towards the density $\muV(x)$ as $\epsilon \to 0^+$, hence we have for $x$ in $\SigmaV$
\begin{equation}\label{densite}
\muV(x)^2 = S(x)^2 \sigma^2(x) = -\frac{1}{(2 \pi)^2} (V'(x)^2 - 4 L(x)).
\end{equation}
This proves that $\muV$ has regularity $C^{\pp-2}$ at any point where it does not vanish. Assuming the form \eqref{regularity} for $S$, we also deduce that the function $S_0$ has regularity at least $C^{\pp - 3 - 2\kk}$ on $\SigmaV$.

Applying \eqref{solvG} on $\R \setminus \Sigma$, we obtain
\begin{equation*}
\frac{1}{2} V'(x) - \int \frac{1}{x-y} d\muV(y)=  \frac{1}{2}  \sqrt{V'(x)^2 - 4  L(x)},
\end{equation*}
and the left-hand side is equal to $\zeta'(x)$. 

Using \eqref{regularity}, \eqref{densite} and the fact that $V$ is regular, we may find a neighborhood $U$ small enough such that $\zeta'$ does not vanish on $U \setminus \SigmaV$ and on which we can write $\zeta'$ as in \eqref{factorization}.
\end{proof}

\subsection{Proof of Lemma \ref{lem32}} \label{sec:preuveF32}
\begin{proof}
We first prove that the image of $F$ is indeed contained in $ C^{1}(U)$. 

For $(t,\psi) = (0,0)$, we have indeed $\Ff(0,0)= \zetaV + c $ and $\zetaV$ is in  $C^{1}(\R)$ by the regularity assumptions on $V$.
 We may also write
\begin{equation*}
\Ff(t,\psi) = \Ff(0,0) - \int \log \frac{|\phi(\cdot) - \phi(y) |}{| \cdot -y|} d\mueq(y) + \frac{1}{2} (V_t \circ \phi - V \circ \phi),
\end{equation*}
and since $\|\psi\|_{C^2(U)} \leq 1/2 $, the second and third terms are  also in  $C^1(U)$. 

Next, we compute the partial derivatives of $\Ff$ at a fixed point $(t_0, \psi_0) \in [-1,1] \times B$.  It is easy to see that 
\begin{equation*}
\frac{\partial \Ff}{\partial t}\Bigr|_{(t_0,\psi_0)} =  \frac{1}{2} \xi \circ \phi_0,
\end{equation*}
and the map $(t_0, \psi_0) \mapsto \xi \circ \phi_0$ is indeed continuous. 

The  Fréchet derivative of $F$ with respect to the second variable can be computed as follows
\begin{multline*}
\Ff(t_0, \psi_0 + \psi_1) = - \int \log \Big| \big(\phi_0(\cdot) - \phi_0(y)\big) +\big(\psi_1(\cdot) - \psi_1(y)\big) \Big|  d\mueq(y) + \frac{1}{2} V_{t_0} \circ(\phi_0 + \psi_1) \\
= \Ff(t_0 , \psi_0) - \int \log \Big| 1 + \frac{\psi_1(\cdot) - \psi_1(y)} {\phi_0(\cdot) - \phi_0(y)} \Big|  d\mueq(y) + \frac{1}{2} \big( V_{t_0} \circ(\phi_0 + \psi_1) - V_{t_0}\circ\phi_0    \big)\\
 = \Ff(t_0 , \psi_0) - \int  \frac{\psi_1(\cdot) - \psi_1(y)} {\phi_0(\cdot) - \phi_0(y)}   d\mueq(y) + \frac{1}{2} \psi_1  V'_{t_0}\circ \phi_0  + \varepsilon_{t_0,\psi_0}(\psi_1) \ ,
\end{multline*}
 where $\varepsilon_{t_0,\psi_0}(\psi_1)$ is given by
\begin{equation*}
\begin{split}
\varepsilon_{t_0,\psi_0}(\psi_1) &=-  \int \bigg[ \log \Big| 1 + \frac{\psi_1(\cdot) - \psi_1(y)} {\phi_0(\cdot) - \phi_0(y)} \Big| - \frac{\psi_1(\cdot) - \psi_1(y)} {\phi_0(\cdot) - \phi_0(y)}   \bigg] d\mueq(y)  \\
 & + \frac{1}{2} \big( V_{t_0} \circ(\phi_0 + \psi_1) - V_{t_0}\circ\phi_0  - \psi_1  V'_{t_0}\circ \phi_0  \big)\ .
\end{split}
\end{equation*}
By differentiating twice inside the integral we get the bound 
\begin{equation*}
\begin{split}
\|\varepsilon_{t_0,\psi_0}(\psi_1)\|_{C^1(U)} \leq C(t_0, \psi_0) \|\psi_1\|^2_{C^2(U)}, 
\end{split}
\end{equation*}
with a constant depending on $V$. It implies that  
\begin{equation*}
\frac{\partial \Ff}{\partial \psi}\Bigr|_{(t_0,\psi_0)} [\psi_1] =-  \int  \frac{\psi_1(\cdot) - \psi_1(y)} {\phi_0(\cdot) - \phi_0(y)}   d\mueq(y) + \frac{1}{2} \psi_1  V'_{t_0}\circ \phi_0 \ ,
\end{equation*}
and we can check that this expression is also continuous in $(t_0, \psi_0)$. In particular, we may observe that
\begin{equation}
\frac{\partial \Ff}{\partial \psi}\Bigr|_{(0,0)} [\psi] = - \XiV[\psi].
\end{equation}

Finally, we prove the bound \eqref{erreur transp}. For any fixed $(t,\psi) \in [-1,1] \times B$, we write 
\begin{equation*}
\begin{split}
\Ff(t,\psi) - \Ff(0,0) = \int_0^1 \frac{d\Ff(st, s\psi)}{ds} ds = \int_0^1 \Big(t \frac{\partial \Ff}{\partial t}\Bigr|_{(st,s\psi)}   + \frac{\partial \Ff}{\partial \psi}\Bigr|_{(st,s\psi)}[\psi]\Big) ds \ ,
\end{split}
\end{equation*}

\noindent we get
\begin{multline}
\|\Ff(t,\psi) - \Ff(0,0) - \frac{t}{2} \xi  + \XiV[\psi] \|_{C^{1}(U)} \leq \int_0^1 \left(\frac{t}{2} \|\xi\circ \phi_s - \xi\|_{C^{1}(U)}\right.  \\ \left.
+ \left\|\frac{\partial \Ff}{\partial \psi}\Bigr|_{(st,s\psi)}[\psi] - \frac{\partial \Ff}{\partial \psi}\Bigr|_{(0,0)}[\psi] \right\|_{C^{1}(U)} \right)ds,
\end{multline}
\noindent with $\phi_s = \id + s \psi$. It  is straightforward to check  that 
\begin{equation*}
\|\xi\circ \phi_s - \xi\|_{C^{1}(U)}  \leq C  \|\xi\|_{C^{2}(U)} \|\psi\|_{C^{1}(U)} \ .
\end{equation*}

\noindent To control the second term inside the integral we write 
\begin{multline*}
\frac{\partial \Ff}{\partial \psi}\Bigr|_{(st,s\psi)} [\psi]- \frac{\partial \Ff}{\partial \psi}\Bigr|_{(0,0)} [\psi]\\  = - \int  \left( \frac{\psi(\cdot) - \psi(y)} {\phi_s(\cdot) - \phi_s(y)} - \frac{\psi(\cdot) - \psi(y)} {\cdot - y}  \right) d\mueq(y) + \frac{1}{2}  \left( V'_{st}\circ \phi_s - V' \right) \psi 
\end{multline*}
 and we obtain
\begin{multline*}
\left\|\frac{\partial \Ff}{\partial \psi}\Bigr|_{(st,s\psi)}[\psi] - \frac{\partial \Ff}{\partial \psi}\Bigr|_{(0,0)}[\psi]\right\|_{C^{1}(U)} \\ 
\leq \int  \left\| \frac{\psi(\cdot) - \psi(y)} {\phi_s(\cdot) - \phi_s(y)} - \frac{\psi(\cdot) - \psi(y)} {\cdot - y}  \right\|_{C^{1}(U)} d\mueq(y) \\ 
+ \left\|  \big( V'_{st}\circ \phi_s - V' \big) \psi\right\|_{C^{1}(U)}
\end{multline*}

\noindent We now use that 
\begin{equation*}
\begin{split}
 \left\|\Big( \frac{\psi(\cdot) - \psi(y)} {\phi_s(\cdot) - \phi_s(y)} - \frac{\psi(\cdot) - \psi(y)} {\cdot - y}  \Big)\right\|_{C^{1}(U)} &=  \left\| \bigg(\frac{\psi(\cdot) - \psi(y)} {\cdot - y} \bigg) \bigg( \frac{\cdot - y} {\phi_s(\cdot) - \phi_s(y)} - 1  \bigg)\right\|_{C^{1}(U)} \\
 &\leq C \|\psi\|_{C^{2}(U)} \left\|\frac{\cdot - y} {\phi_s(\cdot) - \phi_s(y)} - 1\right\|_{C^{1}(U)} \\
 &= C s \|\psi\|_{C^{2}(U)} \left\|\frac{\psi(\cdot) - \psi(y)} {\phi_s(\cdot) - \phi_s(y)} \right\|_{C^{1}(U)} \\
 & \leq C \|\psi\|^2_{C^{2}(U)} \left\|\frac{\cdot - y} {\phi_s(\cdot) - \phi_s(y)}\right\|_{C^{1}(U)} \\
 & \leq C \|\psi\|^2_{C^{2}(U)} \ .
\end{split}
\end{equation*}

\noindent In  the second and the fourth line, we used Leibniz formula . In the last line we used that $s (\psi(\cdot) - \psi(y))/(\cdot - y)$ is uniformely bounded by $1/2$ in $C^2(U)$ so its composition with the function $x\rightarrow 1/(1+x)$ is bounded in $C^2(U)$. We conclude by checking that
\begin{equation*}
\begin{split}
\|  \big( V'_{st}\circ \phi_s - V' \big) \psi\|_{C^{1}(U)} \leq C \Big( \|V\|_{C^{3}(U)} \|\psi\|_{C^{1}(U)} +t \|\psi\|_{C^{2}(U)}\Big)\|\psi\|_{C^{0}(U)} \ .
\end{split}
\end{equation*}
\end{proof}

\subsection{Proof of Lemma \ref{leminv}}
\label{sec:proofinverse}
\begin{proof}
First, we solve the equation $\XiV[\psi] = \hal \xi + c_{\xi}$ in $\SigmaVint$, where $\XiV$ is operator defined in \eqref{Xi}. For $x$ in $\SigmaVint$, we have the following equation
\begin{equation} \label{SDyson}
\frac{V'(x)}{2} = \PV \int \frac{1}{x-y} d\muV(y).
\end{equation}
In particular, for $x$ in $\SigmaVint$, it implies
\begin{equation} \label{XiVasPV}
\XiV[\psi] (x) := \PV \int_{\SigmaV} \frac{\psi(y)}{y-x} \muV(y) dy, 
\end{equation}
and we might thus try to solve 
\begin{equation} \label{singularEqA}
\PV \int_{\SigmaV} \frac{\psi(y)}{y-x} \muV(y) dy =\hal  \xi + \cxi.
\end{equation}
Equation \eqref{singularEqA} is a singular integral equation, we refer to \cite[Chap. 10-11-12]{Muskh} for a detailed treatment. In particular, it is known that if the conditions \eqref{X1} are satisfied, then there exists a solution $\psi_0$ to 
\begin{equation} \label{eqpsizero}
\PV \int_{\SigmaV} \frac{\psi_0(y)}{y-x} dy =\hal  \xi + c_\xi \text{ on } \SigmaVint, 
\end{equation}  
which is explicitly given by the formula
\begin{equation} \label{solPsizero}
\psi_0(x) =  - \frac{\sigma(x)}{2\pi^2} \PV  \int_{\SigmaV}\frac{\xi(y)}{\sigma(y)(y-x)} dy.
\end{equation}
Since we have, for $x$ in $\SigmaVint$
$$
\PV \int_{\SigmaV} \frac{1}{\sigma(y) (y-x)} dy = 0,
$$
we may re-write \eqref{solPsizero} as
\begin{equation} \label{solPsizerob}
\psi_0(x) =  - \frac{\sigma(x)}{2\pi^2}  \int_{\SigmaV}\frac{\xi(y) - \xi(x)}{\sigma(y)(y-x)} dy \text{ on } \SigmaVint,
\end{equation}
where the integral is now a definite Riemann integral. From \eqref{solPsizerob} we deduce that the map $\frac{\psi_0}{\sigma}$ is of class $C^{\rr-1}$ in $\SigmaVint$ and extends readily to a $C^{\rr-1}$ function on $\SigmaV$.

For $d = 0 ,\dots, \rr-1$ and for $x \in \SigmaV$, we compute that
\begin{equation*}
\left( \frac{\psi_0}{\sigma}\right)^{(d)} (x) = - \frac{d!}{2\pi^2} \int_{\SigmaV}\frac{\xi(y) - R_{s_i,d+1}\xi(y)}{\sigma(y)(y-s_i)^{d+1}} dy.
\end{equation*}
In particular, if conditions \eqref{X2} hold, in  view of Lemma \ref{lemmareg} the map 
$$
\psi(x) := \frac{\psi_0(x)}{S(x) \sigma(x)} 
$$
extends to a function of class $(\pp-3-2\kk)\wedge(\rr-1-\kk)$, hence  $C^{2}$ on $\SigmaV$, and in view of \eqref{eqpsizero} it satisfies $\XiV [\psi] = \frac{\xi}{2} + c_\xi$ on $\SigmaV$.

Now, we define $\psi$ outside $\SigmaV$. By definition, for $x$ outside $\SigmaV$, the equation 
$$\XiV[\psi](x) =\hal  \xi(x)+c_\xi$$
 can be written as
 $$ 
 \psi(x)\int \frac {1}{x-y}\, d\muV(y) - \int \frac{\psi(y)} {x-y} \, d\muV(y) - \hal \psi(x) V'(x) = \hal  \xi(x) + c_\xi, 
 $$ 
and thus the choice \eqref{psioutSigma} ensures that $\XiV[\psi] = \hal  \xi + c_\xi$. Moreover, $\psi$ is clearly of class $C^{\rr \wedge (\pp-1)}$ on $\R \setminus \SigmaV$. It remains to check that $\psi$ has the desired regularity at the endpoints of $\SigmaV$. Let us consider $\tilde{\psi}$ an extension of  $\psi$ in $C^l$ with  $l:=(\pp-3-2\kk)\wedge(\rr-1-\kk)$, which coincides with $\psi$ on $\SigmaV$    (given for instance by a Taylor expansion at the endpoints).
 As $\psi$ and $\tilde{\psi}$ are equal on the support  we can rewrite \eqref{psioutSigma} as 
\begin{equation*}
\begin{split}
 \frac{\  \int \frac{\psi(y)}{x-y} d\mueq(y) + \frac{\xi(x)}{2} + c_\xi}{  \int \frac{1}{x-y} d\mueq(y) - \frac{1}{2} V'(x)}  &=  \frac{\  - \int \frac{\tilde{\psi}(x) - \tilde{\psi}(y)}{x-y} d\mueq(y) + \tilde{\psi}(x)\int \frac{1}{x-y} d\mueq(y) + \frac{\xi(x)}{2} + c_\xi}{  \int \frac{1}{x-y} d\mueq(y) - \frac{1}{2} V'(x)}  \\
 &= \tilde{\psi}(x) + \frac{ \frac{\xi(x)}{2} + c_\xi - \XiV[\tilde{\psi}](x)}{  \int \frac{1}{x-y} d\mueq(y) - \frac{1}{2} V'(x)} \ .  
\end{split}
\end{equation*}

Since $\XiV[{\psi}] = \frac{\xi}{2} + c_\xi $ on $\SigmaV$, the numerator on the right hand side of the last equation and its first $l$ derivatives vanish at any endpoint $\alpha$. From Lemma \eqref{lemmareg} we conclude that $\psi$ is of class  $l-\kk= (\pp-3-3\kk)\wedge(\rr-1-2\kk)$ at $\alpha$, hence $C^2$ from \eqref{conditionpr}.

\end{proof}

\subsection{Proof of Lemma \ref{lem42}}
\begin{proof}
The first item is a consequence of the fact that $\muV$ minimizes the logarithmic potential energy \eqref{energy} and hence 
as is well-known $\int - \log |\cdot - y|\, d\muV(y)+ \hal V$ is constant on the support of $\muV$. Differentiating this and integrating against $h d\muV$ gives the result. 
For the second relation, by definition of $\psi$ we have 
\begin{equation*}
\frac{\xi}{2} + c_{\xi} = \int \frac{\psi(x) - \psi(y)}{x-y} d\muV(y) - \frac{1}{2} \psi V' \ ,
\end{equation*}
 and thus 
\begin{equation*}
{\xi'} = 2 \int \frac{\psi(y) - \psi(x) - \psi'(x)(y-x)}{(x-y)^2}d\muV(y) -  \psi' V' -  \psi V''  \ .
\end{equation*}
 Integrating both sides against $\psi \muV$ yields 
\begin{multline*}
\int {\xi'} \psi  d\mueq= 2 \iint \frac{(\psi(y) - \psi(x) - \psi'(x)(y-x))\psi(x)}{(x-y)^2}d\muV(y)d\muV(x) \\ 
- \int \psi \psi' V' d\mueq- \int   V'' \psi^2 d\mueq .
\end{multline*}
 Using \eqref{equil} for the second term we obtain 
\begin{multline*}
\int {\xi'} \psi d\mueq  = 2 \iint \frac{(\psi(y) - \psi(x) - \psi'(x)(y-x))\psi(x)}{(y-x)^2}d\muV(y)d\muV(x) \\ 
- \iint \frac{\psi \psi'(y) - \psi \psi'(x) }{y-x} d\muV(x) d\muV(y) - \int   V'' \psi^2 d\muV .
\end{multline*}
We may then combine the first two terms in the right-hand side to obtain \eqref{identiteaprouver}.
\end{proof}

\bibliographystyle{alpha}
\bibliography{CLTBeta_Ensembles}

\newcommand{\etalchar}[1]{$^{#1}$}
\begin{thebibliography}{BdMPS95}

\bibitem[AKM17]{Armstrong:2017qy}
S.~Armstrong, T.~Kuusi, and J.-C. Mourrat.
\newblock Quantitative stochastic homogenization and large-scale regularity.
\newblock {\em https://arxiv.org/abs/1705.05300}, 2017.

\bibitem[BdMPS95]{BdMS}
A.~Boutet~de Monvel, L.~Pastur, and M.~Shcherbina.
\newblock On the statistical mechanics approach in the random matrix theory:
  integrated density of states.
\newblock {\em Journal of statistical physics}, 79(3):585--611, 1995.

\bibitem[BEY14]{bourgade2014}
P.~Bourgade, L.~Erd{\H o}s, L.s, and H.-T. Yau.
\newblock Universality of general $\beta$ -ensembles.
\newblock {\em Duke Math. J.}, 163(6):1127--1190, 04 2014.

\bibitem[BFG13]{BFG}
F.~Bekerman, A.~Figalli, and A.~Guionnet.
\newblock Transport maps for $\beta$-matrix models and universality.
\newblock {\em Communications in Mathematical Physics}, 338(2):589--619, 2013.

\bibitem[BG13a]{BorGui2}
G.~Borot and A.~Guionnet.
\newblock Asymptotic expansion of $\beta$ matrix models in the multi-cut
  regime.
\newblock {\em arXiv preprint arXiv:1303.1045}, 2013.

\bibitem[BG13b]{BorGui1}
G.~Borot and A.~Guionnet.
\newblock Asymptotic expansion of $\beta$ matrix models in the one-cut regime.
\newblock {\em Communications in Mathematical Physics}, 317(2):447--483, 2013.

\bibitem[BL16]{bekerman2016}
F.~Bekerman and A.~Lodhia.
\newblock Mesoscopic central limit theorem for general $\beta $-ensembles.
\newblock {\em arXiv preprint arXiv:1605.05206}, 2016.

\bibitem[CK06]{claeys2006universality}
T.~Claeys and A.B.J. Kuijlaars.
\newblock Universality of the double scaling limit in random matrix models.
\newblock {\em Communications on Pure and Applied Mathematics},
  59(11):1573--1603, 2006.

\bibitem[CKI10]{claeys2010higher}
T.~Claeys, I.~Krasovsky, and A.~Its.
\newblock Higher-order analogues of the {T}racy-{W}idom distribution and the
  {P}ainlev{\'e} ii hierarchy.
\newblock {\em Communications on pure and applied mathematics}, 63(3):362--412,
  2010.

\bibitem[Cla08]{claeys2008birth}
T.~Claeys.
\newblock Birth of a cut in unitary random matrix ensembles.
\newblock {\em International Mathematics Research Notices}, 2008:rnm166, 2008.

\bibitem[DKM98]{deift1998new}
P.~Deift, T.~Kriecherbauer, and K.~T.-R. McLaughlin.
\newblock New results on the equilibrium measure for logarithmic potentials in
  the presence of an external field.
\newblock {\em Journal of approximation theory}, 95(3):388--475, 1998.

\bibitem[DKM{\etalchar{+}}99]{DKM}
P.~Deift, T.~Kriecherbauer, K.~T.-R. McLaughlin, S.~Venakides, and X.~Zhou.
\newblock Uniform asymptotics for polynomials orthogonal with respect to
  varying exponential weights and applications to universality questions in
  random matrix theory.
\newblock {\em Communications on Pure and Applied Mathematics},
  52(11):1335--1425, 1999.

\bibitem[For10]{forrester}
P.~Forrester.
\newblock {\em Log-gases and random matrices}.
\newblock Princeton University Press, 2010.

\bibitem[GMS07]{guionnet2007second}
A.~Guionnet and E.~Maurel-Segala.
\newblock Second order asymptotics for matrix models.
\newblock {\em The Annals of Probability}, 35(6):2160--2212, 2007.

\bibitem[GS14]{guionnet2014free}
A.~Guionnet and D.~Shlyakhtenko.
\newblock Free monotone transport.
\newblock {\em Inventiones mathematicae}, 197(3):613--661, 2014.

\bibitem[Joh98]{Johansson}
K.~Johansson.
\newblock On fluctuations of eigenvalues of random {H}ermitian matrices.
\newblock {\em Duke Mathematical Journal}, 91(1):151--204, 1998.

\bibitem[KM00]{kuijlaars2000generic}
A.~Kuijlaars and K.~McLaughlin.
\newblock Generic behavior of the density of states in random matrix theory and
  equilibrium problems in the presence of real analytic external fields.
\newblock {\em Communications on Pure and Applied Mathematics}, 53(6):736--785,
  2000.

\bibitem[LLW17]{Lambert:2017kq}
G.~Lambert, M.~Ledoux, and C.~Webb.
\newblock Stein's method for normal approximation of linear statistics of
  beta-ensembles.
\newblock {\em https://arxiv.org/abs/1706.10251}, 06 2017.

\bibitem[LS15]{ls1}
T.~Lebl{\'e} and S.~Serfaty.
\newblock Large deviation principle for empirical fields of {L}og and {R}iesz
  gases.
\newblock {\em To appear in Inventiones Mathematicae}, 2015.

\bibitem[LS16]{ls2}
T.~Lebl{\'e} and S.~Serfaty.
\newblock Fluctuations of two-dimensional {C}oulomb gases.
\newblock {\em arXiv preprint arXiv:1609.08088}, 2016.

\bibitem[MdMS14]{MR3201924}
M.~Ma\"\i~da and {\'E}.~Maurel-Segala.
\newblock Free transport-entropy inequalities for non-convex potentials and
  application to concentration for random matrices.
\newblock {\em Probab. Theory Related Fields}, 159(1-2):329--356, 2014.

\bibitem[Mo08]{mo2008riemann}
M.Y. Mo.
\newblock The {R}iemann--{H}ilbert approach to double scaling limit of random
  matrix eigenvalues near the "birth of a cut" transition.
\newblock {\em International Mathematics Research Notices}, 2008.

\bibitem[Mus92]{Muskh}
N.~I. Muskhelishvili.
\newblock {\em Singular integral equations}.
\newblock Dover Publications, Inc., New York, 1992.

\bibitem[PS14]{PS15}
M.~Petrache and S.~Serfaty.
\newblock Next order asymptotics and renormalized energy for riesz
  interactions.
\newblock {\em Journal of the Institute of Mathematics of Jussieu}, pages
  1--69, 2014.

\bibitem[Shc13]{scherbi1}
M.~Shcherbina.
\newblock Fluctuations of linear eigenvalue statistics of $\beta$ matrix models
  in the multi-cut regime.
\newblock {\em Journal of Statistical Physics}, 151(6):1004--1034, 2013.

\bibitem[Shc14]{shchange}
M.~Shcherbina.
\newblock Change of variables as a method to study general $\beta$-models: bulk
  universality.
\newblock {\em Journal of Mathematical Physics}, 55(4):043504, 2014.

\bibitem[SS15]{SS15a}
{\'E}.~Sandier and S.~Serfaty.
\newblock 1{D} log gases and the renormalized energy: crystallization at
  vanishing temperature.
\newblock {\em Probability Theory and Related Fields}, 162(3-4):795--846, 2015.

\bibitem[ST13]{ST}
E.B. Saff and V.~Totik.
\newblock {\em Logarithmic potentials with external fields}, volume 316.
\newblock Springer Science \& Business Media, 2013.

\end{thebibliography}

%The operator $\XiV$ is well-known in this context, and was an essential tool in \cite{BFG} for proving universality for $\beta$-ensembles with non-critical one cut potentials by transport methods.

\end{document}